\documentclass[11pt,a4paper]{article}%
\usepackage{jheppub}
\usepackage{amsmath}
\usepackage{graphicx}
\usepackage{amsfonts}
\usepackage{amssymb}%
\newtheorem{theorem}{Theorem}

\newtheorem{lemma}[theorem]{Lemma}
\newtheorem{remark}[theorem]{Remark}

\newenvironment{proof}[1][Proof]{\textbf{#1.} }{\ \rule{0.5em}{0.5em}}

\title{\boldmath Bethe Ansatz and exact form factors of the $O(N)$ Gross Neveu-model}

\author[a]{Hrachya M. Babujian}
\author[b]{Angela Foerster}
\author[c]{and Michael Karowski}

\affiliation[a]{Yerevan Physics Institute,\\
Alikhanian Brothers 2, Yerevan, 375036 Armenia and\\
International Institute of Physics,
Universidade Federal do Rio Grande do Norte (UFRN),\\
59078-400 Natal-RN, Brazil}
\affiliation[b]{Instituto de F\'{\i}sica da UFRGS,\\
Av. Bento Gon\c{c}alves 9500, Porto Alegre, RS - Brazil}
\affiliation[c]{Institut f\"{u}r Theoretische Physik, FU-Berlin,\\
Arnimallee 14, 14195 Berlin, Germany}

\emailAdd{babujian@yerphi.am}
\emailAdd{angela@if.ufrgs.br}
\emailAdd{karowski@physik.fu-berlin.de}

\abstract{
We apply previous results on the $O(N)$ Bethe Ansatz \cite{BFK5,BFK6,BFK7} to
construct a general form factor formula for the $O(N)$ Gross-Neveu model. We
examine this formula for several operators, such as the energy momentum, the
spin-field and the current. We also compare these results with the $1/N$
expansion of this model and obtain full agreement. We discuss bound state form
factors, in particular for the three particle form factor of the field. In
addition for the two particle case we prove a recursion relation for the
K-functions of the higher level Bethe Ansatz.
\\[1cm]
\textsc{Keywords}: Exact S-Matrix, Form Factors, Bethe Ansatz,
Integrable Field Theories
}

\begin{document}

\maketitle

\section{Introduction}

The $O(N)~\sigma$- and Gross-Neveu (GN) models are integrable and
asymptotically free quantum field theories in 1+1 dimension. The S-matrices of
these two models correspond to two solutions of the Yang-Baxter equation
\cite{ZZ3,ZZ4}. In previous articles we constructed the $O(N)$ nested
off-shell Bethe Ansatz \cite{BFK5,BFK6} and applied this technique to
construct the exact form factors for the $O(N)~\sigma$ model \cite{BFK7}. Here
we extend this work and construct the form factors for the $O(N)$ Gross-Neveu
model for arbitrary number of fundamental particles (for the two-particle case
see \cite{KW}). The model exhibits a very rich bound state structure and kinks
(see e.g. \cite{KT1}), turning this study even more challenging.

Before we recall the S-matrix and all other details of this model we should
mention that the integrable structure present in 1+1 dimension is now becoming
relevant and actual in higher dimensional gauge theories under specific
circumstances. Remarkably, in the articles \cite{BSV1,BSV2,BSV3,BSV4} (see
also references therein) a non-perturbative formulation of planar scattering
in the $N=4$ Supersymmetric Yang-Mills theory (SYM) with the so called
polygonal Wilson loops was proposed and a new decomposition of the Wilson
loops in terms of the fundamental building blocks-Pentagon transitions was
introduced. These transitions are directly related to the dynamics of the
Gubser-Klebanov-Polyakov flux-tube \cite{GKP}, which can be computed exactly
by exploring the integrability. In addition, three axioms about the
transitions that single particles must satisfy were postulated and,
interestingly, it is possible to verify that these axioms correspond to some
deformations of the form factor equations in $1+1$- dimensional integrable
quantum field theories. Such exact and constructive developments in the $N=4$
SYM theory opens, indeed, large perspectives in the view of using the exact
integrability and the full machinery of the form factor program to get
physical insights, specially in the case of non-trivial symmetry groups, such
as $SU(N)$ and $O(N)$.

In this article we consider the $O(N)$-Gross-Neveu model for $N=$ even. We do
not use any Lagrangian to construct the model, nevertheless, we give the
following motivation. The $O(N)$-Gross-Neveu model describes the interaction
of $N/2$ Dirac (or $N$ Majorana) fermions defined by the
Lagrangian\footnote{The Lagrangian (\ref{LGN}) is invariant under $O(N)$
transformations of the vector of $N$ Majorana fermi fields $\psi_{\alpha
}^{(i)}~(\alpha=1,\dots,N/2),$~$i=1,2,$ where $\psi_{\alpha}=\psi_{\alpha
}^{(1)}+i\psi_{\alpha}^{(2)}$ \cite{ZZ4}.} \cite{GN}%
\begin{equation}
\mathcal{L}^{GN}=\sum_{\alpha=1}^{N/2}\bar{\psi}_{\alpha}i\gamma\partial
\psi_{\alpha}+\frac{1}{2}\,g^{2}\left(  \sum_{i=1}^{N/2}\bar{\psi}_{\alpha
}\psi_{\alpha}\right)  ^{2}\,. \label{LGN}%
\end{equation}
It is known from semi-classical calculations \cite{DHN} that there are bound
states of two fundamental fermions in the scalar and the anti-symmetric tensor
channel. Furthermore there are kinks such that the fundamental fermions are
kink-kink bound states. The bootstrap program does not use the Lagrangian, but
we are looking for an factorizing S-matrix of an $O(N)$-isovector $N$-plett of
self conjugate fundamental fermions. However, now we assume bound states in
the scalar and anti-symmetric tensor channel of two of them.

In this article we use the techniques of \cite{BFK5,BFK7} to construct the
form factors of the $O(N)$-Gross-Neveu model. We apply the general results to
compute exact form factors for the energy-momentum, the spin-field and the
current. The exact results are compared with the ones obtained in perturbation
theory using the $1/N$ expansion. The final aim of the form factor program is
to obtain explicit results for the correlation functions or Wightman functions
in the framework of 2-dimensional integrable QFTs. In \cite{KW,BKW} the
concept of generalized form factors was introduced and developed further by
Smirnov \cite{Sm}. We call the matrix elements of fields with many particle
states: \textquotedblleft generalized form factors\textquotedblright. Matrix
difference equations (the generalized Watson's equations) are solved by using
the \textquotedblleft off-shell Bethe Ansatz" \cite{BKZ2,BFKZ,BFK1,BFK5},
which was introduced in \cite{B3} to solve the Knizhnik-Zamolodchikov
equations. Other approaches to form factors in integrable quantum field
theories can be found in \cite{CM2,FMS,YZ,Lu0,Lu,Lu1,BL,LuZ,Or}. For articles
considering the form factor program for Bethe Ansatz solvable models with
nesting see also
\cite{Pozsgay:2012wu,Pakuliak:2015qga,Pakuliak:2015fma,1751-8121-48-43-435001}.

The general form factor formula in terms of an integral representation is the
main result of this paper. It solves the form factors equations. The matrix
element of a local operator $\mathcal{O}(x)$ for a state of $n$ particles of
kind $\alpha_{i}$ with rapidities $\theta_{i}$
\begin{equation}
\langle\,0\,|\,\mathcal{O}(x)\,|\,\theta_{1},\dots,\theta_{n}\,\rangle
_{\underline{\alpha}}^{in}=e^{-ix(p_{1}+\cdots+p_{n})}F_{\underline{\alpha}%
}^{\mathcal{O}}(\underline{\theta}) \label{2.8}%
\end{equation}
defines the generalized form factor $F_{\underline{\alpha}}^{\mathcal{O}%
}(\underline{\theta})$. Here we restrict $\alpha$ to the fundamental particles
of the model, which form an isovector $N$-plett of $O(N)$. Following \cite{KW}
we write%
\begin{equation}
F_{\underline{\alpha}}^{\mathcal{O}}(\underline{\theta})=K_{\underline{\alpha
}}^{\mathcal{O}}(\underline{\theta})\prod_{1\leq i<j\leq n}F(\theta_{ij})
\label{2.10}%
\end{equation}
where $F(\theta)$ is the minimal form factor function.

For the K-function we propose the same Ansatz as for the $\sigma$-model in
\cite{BFK7} in terms of a nested `off-shell' Bethe Ansatz
\begin{equation}
\fbox{$\rule[-0.2in]{0in}{0.5in}\displaystyle~K_{\underline{\alpha}%
}^{\mathcal{O}}(\underline{\theta})=N_{n}^{\mathcal{O}}\int_{\mathcal{C}%
_{\underline{\theta}}^{(1)}}dz_{1}\cdots\int_{\mathcal{C}_{\underline{\theta}%
}^{(m)}}dz_{m}\,\tilde{h}(\underline{\theta},\underline{z})\,p^{\mathcal{O}%
}(\underline{\theta},\underline{z})\,\tilde{\Psi}_{\underline{\alpha}%
}(\underline{\theta},\underline{z})$\thinspace.} \label{BA}%
\end{equation}
Here $\tilde{h}(\underline{\theta},\underline{z})$ is a scalar function which
depends only on the S-matrix. The scalar p-function $p^{\mathcal{O}%
}(\underline{\theta},\underline{z})$ which is in general a simple function of
$e^{\theta_{i}}$ and $e^{z_{j}}$ depends on the specific operator
$\mathcal{O}(x)$. This Ansatz transforms the complicated form factor matrix
equations (see (\ref{1.10})-(\ref{1.14}) below) to simple equations for the
scalar function $p^{\mathcal{O}}(\underline{\theta},\underline{z})$ (see also
\cite{BFK1}). The integration contour $\mathcal{C}_{\underline{\theta}}$ will
be specified in section \ref{s4}. The state $\tilde{\Psi}_{\underline{\alpha}%
}$ in (\ref{BA}) is a linear combination of the basic Bethe Ansatz co-vectors
(see \cite{BFK7} and (\ref{BS}))%
\begin{equation}
\tilde{\Psi}_{\underline{\alpha}}(\underline{\theta},\underline{z}%
)=L_{\underline{\mathring{\beta}}}(\underline{z})\,\tilde{\Phi}_{\underline
{\alpha}}^{\underline{\mathring{\beta}}}(\underline{\theta},\underline{z})\,.
\label{PSI}%
\end{equation}
The nested off-shell Bethe Ansatz is obtained by making for $L_{\underline
{\mathring{\beta}}}(\underline{z})$ an Ansatz like (\ref{BA}) and iterating
this procedure. In the present paper we mainly consider the case where
$\alpha$ correspond to the fundamental fermions of the $O(N)$-Gross-Neveu
model Lagrangian (\ref{LGN}). In forthcoming publications we will consider the
kinks \cite{BFK11} and we will discuss, in particular, the $O(6)$-Gross-Neveu
model in more detail \cite{BFK10}.

The `off-shell' Bethe Ansatz states are highest weight states if they satisfy
certain matrix difference equations (see for instance \cite{BFK5}). For $n$
particle states the $O(N)$ weights are%
\[
(w_{1},\dots,w_{N/2})=\left(  n-n_{1},\dots,n_{N/2-2}-n_{-}-n_{+},n_{-}%
-n_{+}\right)
\]
where $n_{1}=m,n_{2},\dots$ are the numbers of integrations in (\ref{BA}) and
the higher levels of the nesting. In particular $n_{\pm}$ are the numbers of
positive/negative chirality spinors. The various levels of the nested Bethe
Ansatz correspond to the nodes of the Dynkin diagram of the corresponding Lie
algebra (see for instance \cite{OWR,Re1,RW} and references therein). Here we
have $D_{N/2}$ for $N=$ even (see Fig. \ref{f2.2}). In \cite{BFK7} we used for
the $O(N)$~$\sigma$-model the group isomorphy $O(4)\simeq SU(2)\otimes SU(2)$
to start the nesting procedure with form factors of the $SU(2)$ chiral
Gross-Neveu model \cite{BFK4}. For the $O(N)$ Gross-Neveu model it is also
possible to use the group isomorphy $O(6)\simeq SU(4)$ to start the nesting
with form factors of the $SU(4)$ chiral Gross-Neveu model \cite{BFK4}. This
will be performed in detail in a separate paper \cite{BFK10}. For the on-shell
Bethe Ansatz for $N$ even see also \cite{dVK}.

Section \ref{s2} provides some known results and the notation for the $O(N)$
Gross-Neveu S-matrix, the bound states, etc. In Section \ref{s3} we recall the
general form factor equations and obtain the minimal form factor function. In
Section \ref{s4} we present the general exact form factors formula for the
$O(N)$-Gross-Neveu model and discuss the higher levels of the nested off-shell
Bethe Ansatz. In section \ref{s5} the general results are applied to some
examples. The more complicated proofs and calculations are delegated to the appendices.

\section{General settings\label{s2}}

\subsection{The $O(N)$-Gross-Neveu S-matrix}

We consider the fundamental particles of the Lagrangian (\ref{LGN}) which are
fermions and transform as the vector representation of $O(N)$. The structure
of the S-matrix is the same as that of the nonlinear $\sigma$-model
\cite{BFK7} however, here we are looking for a solution of the $O(N)$%
-Yang-Baxter equations with a bound state pole in the physical strip
$0<\operatorname{Im}\theta<\pi$. Therefore, here \textquotedblleft
minimality\textquotedblright\ implies that the S-matrix for the scattering of
two fundamental particles is of the form%
\begin{equation}
S(\theta,N)=\frac{\sinh\theta+i\sin\pi\nu}{\sinh\theta-i\sin\pi\nu}S^{\min
}(\theta),~~\text{with }\nu=\frac{2}{N-2}. \label{SGN}%
\end{equation}
This S-matrix\ was given by Zamolodchikov-Zamolodchikov \cite{ZZ4}. The first
factor in (\ref{SGN}) is the sine-Gordon breather-breather \cite{KT} amplitude
and $S^{\min}$ is the minimal $O(N)$ S-matrix which is the one of the
nonlinear $\sigma$-model (see e.g. \cite{BFK7}). The position of the pole is
dictated by the condition \cite{K1} that the pole has to be cancelled by a
zero in the amplitude $S_{+}^{\min}$. This condition\footnote{An additional
pole in $S_{+}^{GN}$ would contradict positivity in the Hilbert space (for
details see \cite{K1}).} fixes the pole and therefore the bound state mass
spectrum
\begin{equation}
m_{k}=2m\sin\tfrac{1}{2}k\pi\nu\quad(k=1,2,\dots,N/2-2)\,. \label{mk}%
\end{equation}
For each \textquotedblleft principal\textquotedblright\ quantum number $k$
there exist particles $b_{k}^{(r)}$ which are anti-symmetric tensors of rank
$r=k,\,k-2,\dots\geq0$, i.e.~they transform according to the $r$-th
fundamental representation of $O(N)$. These particles are bosons/fermions
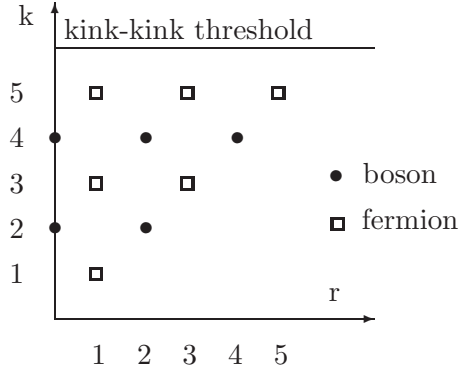
\begin{figure}[th]%
\[
\unitlength=6mm
\begin{picture}(10,8) \put(1,1){\vector(1,0){7}} \put(1,1){\vector(0,1){7}} \put(1,7){\line(1,0){7}} \put(1.2,7.2){kink-kink threshold} \put(7,1.4){r} \put(.2,7.5){k} \put(0,1.8){1} \put(0,2.8){2} \put(0,3.8){3} \put(0,4.8){4} \put(0,5.8){5} \put(1.8,0){1} \put(2.8,0){2} \put(3.8,0){3} \put(4.8,0){4} \put(5.8,0){5} \thicklines \put(2,2){\makebox(0,0){\framebox(.2,.2)}} \multiput(2,4)(2,0){2}{\makebox(0,0){\framebox(.2,.2)}} \multiput(2,6)(2,0){3}{\makebox(0,0){\framebox(.2,.2)}} \multiput(1,3)(2,0){2}{\makebox(0,0){$\bullet$}} \multiput(1,5)(2,0){3}{\makebox(0,0){$\bullet$}} \put(7,4){$\bullet$~~boson} \put(7.1,3){\framebox(.2,.2)~~~fermion} \end{picture}
\]
\caption{Particle spectrum of the $O(N)$-Gross-Neveu model for $N=14$}%
\label{f2.1}%
\end{figure}for $k$ even/odd. In addition there exist \textquotedblleft
kinks\textquotedblright\ of mass $m$ which transform as the two spinor
representations of $O(N)$ (with positive or negative isotopic chirality).

Note the intimate connection between the spectrum of the GN-model, figure
\ref{f2.1}, and the Dynkin diagram figure \ref{f2.2}. There exist
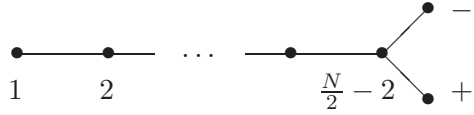
\begin{figure}[th]%
\[
\unitlength=6mm\begin{picture}(10,2) \put(0,1){\line(1,0){3}} \multiput(0,1)(2,0){2}{\makebox(0,0){$\bullet$}} \put(3.6,.94){$\dots$} \put(5,1){\line(1,0){3}} \multiput(6,1)(2,0){2}{\makebox(0,0){$\bullet$}} \put(8,1){\line(1,1){1}} \put(8,1){\line(1,-1){1}} \multiput(9,0)(0,2){2}{\makebox(0,0){$\bullet$}} \put(-.2,0){1} \put(1.8,0){2} \put(6.6,0){$\frac N2-2$} \put(9.5,0){$+$} \put(9.5,1.8){$-$} \end{picture}
\]
\caption{Dynkin diagram for $O(N)$}%
\label{f2.2}%
\end{figure}exclusively such one-particle states which transform according to
one of the fundamental (or trivial) representations of $O(N)$.

For the Bethe Ansatz it is convenient as in \cite{BFK7} to use instead of the
real basis\newline$|\alpha\rangle_{r}\,,$~$\left(  \alpha=1,2,\dots,N\right)
$ the complex basis $|1\rangle,|2\rangle,\dots,|\bar{2}\rangle,|\bar{1}%
\rangle$%
\[
\left.
\begin{array}
[c]{ccc}%
|\alpha\rangle & = & \frac{1}{\sqrt{2}}\left(  |2\alpha-1\rangle_{r}%
+i|2\alpha\rangle_{r}\right)  \\
|\bar{\alpha}\rangle & = & \frac{1}{\sqrt{2}}\left(  |2\alpha-1\rangle
_{r}-i|2\alpha\rangle_{r}\right)
\end{array}
\right\}  ~,~~\alpha=1,2,\dots,N/2\,.
\]
Then the S-matrix writes in terms of the components as%
\begin{equation}
S_{\alpha\beta}^{\delta\gamma}(\theta)=b(\theta)\delta_{\alpha}^{\gamma}%
\delta_{\beta}^{\delta}+c(\theta)\delta_{\alpha}^{\delta}\delta_{\beta
}^{\gamma}+d(\theta)\mathbf{C}^{\delta\gamma}\mathbf{C}_{\alpha\beta}\label{S}%
\end{equation}
with the rapidity difference $\theta$ of the particles and the
\textquotedblleft charge conjugation matrices\textquotedblright\ \
\begin{equation}
\mathbf{C}_{\alpha\beta}=\delta_{\alpha\bar{\beta}}~\text{and }\mathbf{C}%
^{\alpha\beta}=\delta^{\alpha\bar{\beta}}.\label{C}%
\end{equation}
The Yang-Baxter-, crossing- and unitarity-relation write as in \cite{BFK7}.
The highest weight amplitude is $a(\theta)=S_{+}(\theta)=b(\theta)+c(\theta)$%
\begin{align}
a(\theta) &  =\exp\left(  2\int_{0}^{\infty}\frac{dt}{t}\left(  \frac
{e^{-t\left(  1-\nu\right)  }-e^{-t}}{1+e^{-t}}\right)  \sinh t\frac{\theta
}{i\pi}\right)  \label{a}\\
&  =\frac{\Gamma\left(  1-\frac{1}{2\pi i}\theta\right)  \Gamma\left(
\frac{1}{2}+\frac{1}{2\pi i}\theta\right)  }{\Gamma\left(  1+\frac{1}{2\pi
i}\theta\right)  \Gamma\left(  \frac{1}{2}-\frac{1}{2\pi i}\theta\right)
}\frac{\Gamma\left(  1-\frac{1}{2}\nu+\frac{1}{2\pi i}\theta\right)
\Gamma\left(  \frac{1}{2}-\frac{1}{2}\nu-\frac{1}{2\pi i}\theta\right)
}{\Gamma\left(  1-\frac{1}{2}\nu-\frac{1}{2\pi i}\theta\right)  \Gamma\left(
\frac{1}{2}-\frac{1}{2}\nu+\frac{1}{2\pi i}\theta\right)  }\nonumber
\end{align}
with $\nu=2/(N-2)$. For later convenience we introduce
\begin{equation}
\tilde{S}_{\alpha\beta}^{\delta\gamma}(\theta)=S_{\alpha\beta}^{\delta\gamma
}(\theta)/S_{+}(\theta)=\tilde{b}(\theta)\delta_{\alpha}^{\gamma}\delta
_{\beta}^{\delta}+\tilde{c}(\theta)\delta_{\alpha}^{\delta}\delta_{\beta
}^{\gamma}+\tilde{d}(\theta)\mathbf{C}^{\delta\gamma}\mathbf{C}_{\alpha\beta
}\label{1.1}%
\end{equation}
with%
\begin{equation}
\tilde{b}(\theta)=\frac{\theta}{\theta-i\pi\nu},~\tilde{c}(\theta
)=\mathbf{-}\frac{i\pi\nu}{\theta-i\pi\nu},~\tilde{d}(\theta)=\mathbf{-}%
\frac{\theta}{\theta-i\pi\nu}\frac{i\pi\nu}{i\pi-\theta}\,.\label{1.2}%
\end{equation}
We will also need $\mathring{S}(z)$ the S-matrix for $O(N-2)$
\begin{equation}
\tilde{\mathring{S}}_{\alpha\beta}^{\delta\gamma}(\theta)=\mathring{S}%
_{\alpha\beta}^{\delta\gamma}(\theta)/\mathring{S}_{+}(\theta)=\tilde
{\mathring{b}}(\theta)\delta_{\alpha}^{\gamma}\delta_{\beta}^{\delta}%
+\tilde{\mathring{c}}(\theta)\delta_{\alpha}^{\delta}\delta_{\beta}^{\gamma
}+\tilde{\mathring{d}}(\theta)\mathbf{C}^{\delta\gamma}\mathbf{C}_{\alpha
\beta}\label{1.1a}%
\end{equation}
where $\nu$ is replaced by $\mathring{\nu}=2/(N-4)$.

\subsubsection{Bound states}

Following \cite{K1,KT1,BK} we write for the fundamental fermions $\alpha
,\beta,\beta^{\prime},\alpha^{\prime}$
\begin{equation}
i\operatorname*{Res}_{\theta=i\eta}\left(  \sigma S\right)  _{\alpha\beta
}^{\beta^{\prime}\alpha^{\prime}}(\theta)=\sum_{\gamma}\Gamma_{\gamma}%
^{\beta^{\prime}\alpha^{\prime}}\Gamma_{\alpha\beta}^{\gamma}:\quad%
\begin{array}
[c]{l}%
\unitlength3mm\begin{picture}(12,8) \put(0,3.5){$i\operatorname*{Res}$} \put(4,2){\line(1,2){2}} \put(6,2){\line(-1,2){2}} \put(3.4,.5){$\alpha$} \put(5.6,.5){$\beta$} \put(5.6,6.5){$\alpha'$} \put(3.4,6.5){$\beta'$} \put(7.7,3.7){$=$} \put(9.4,-.5){$\alpha$} \put(11.6,-.5){$\beta$} \put(11.6,7.5){$\alpha'$} \put(9.4,7.5){$\beta'$} \put(11,1){\oval(2,4)[t]} \put(11,3){\line(0,1){2}} \put(11,3){\makebox(0,0){$\bullet$}} \put(11,5){\makebox(0,0){$\bullet$}} \put(11,7){\oval(2,4)[b]} \end{picture}
\end{array}
\label{b1.40}%
\end{equation}
where $\sigma=-1$ is the statistics factor. The intertwiner $\Gamma
_{\alpha\beta}^{\gamma}$ and the dual one $\Gamma_{\gamma}^{\beta\alpha}$
satisfy the crossing relation
\begin{equation}
\Gamma_{\gamma}^{\beta\alpha}=\mathbf{C}_{\gamma\gamma^{\prime}}\Gamma
_{\alpha^{\prime}\beta^{\prime}}^{\gamma^{\prime}}\mathbf{C}^{\beta^{\prime
}\beta}\mathbf{C}^{\alpha^{\prime}\alpha}:\quad%
\begin{array}
[c]{l}%
\unitlength2.5mm\begin{picture}(14,8) \put(1,7){\oval(2,6)[b]} \put(1,1){\line(0,1){3}} \put(1,4){\makebox(0,0){$\bullet$}} \put(.6,0){$\gamma$} \put(-.4,7.5){$\beta$} \put(1.6,7.5){$\alpha$} \put(3,4){$=$} \put(6,1){\line(0,1){4}} \put(7.5,5){\oval(3,4)[t]} \put(9,5){\makebox(0,0){$\bullet$}} \put(9,4){\oval(2,2)[t]} \put(11,4){\oval(2,2)[b]} \put(11,4){\oval(6,6)[b]} \put(12,4){\line(0,1){3}} \put(14,4){\line(0,1){3}} \put(5.6,0){$\gamma$} \put(11.6,7.5){$\beta$} \put(13.6,7.5){$\alpha$} \end{picture}
\end{array}
\label{b1.43}%
\end{equation}
with the charge conjugation matrix $\mathbf{C}$ (\ref{C}). Here we have for
$\eta=\pi\nu$%
\begin{align*}
i\operatorname*{Res}_{\theta=i\pi\nu}\left(  \sigma S\right)  _{\alpha\beta
}^{\beta^{\prime}\alpha^{\prime}}(\theta)  &  =-i\operatorname*{Res}%
_{\theta=i\pi\nu}\frac{\sinh\theta+i\sin\pi\nu}{\sinh\theta-i\sin\pi\nu
}\left(  S^{\min}\right)  _{\alpha\beta}^{\beta^{\prime}\alpha^{\prime}%
}(\theta)\\
&  =2\tan\pi\nu\,\left(  S^{\min}\right)  _{\alpha\beta}^{\beta^{\prime}%
\alpha^{\prime}}(i\pi\nu)\,.
\end{align*}

\section{Generalized form factors}

\paragraph{Form factor equations}

\label{s3}

The form factor equations for the $O(N)$ Gross-Neveu-model are similar to the
ones of the $O(N)~\sigma$-model in \cite{BFK7}. However, here there are
additional statistics factors. The $F_{\underline{\alpha}}^{\mathcal{O}%
}(\underline{\theta})$ defined by (\ref{2.8}) are considered as the components
of a co-vector valued function $F_{1\dots n}^{\mathcal{O}}(\underline{\theta
})$ which satisfies:

\begin{itemize}
\item[(i)] Watson's equation
\begin{equation}
F_{\dots ij\dots}^{\mathcal{O}}(\dots,\theta_{i},\theta_{j},\dots)=F_{\dots
ji\dots}^{\mathcal{O}}(\dots,\theta_{j},\theta_{i},\dots)\,\left(  \sigma
S\right)  _{ij}(\theta_{ij})\label{1.10}%
\end{equation}
with $\theta_{ij}=\theta_{i}-\theta_{j}$ and $\sigma_{ij}=-1$ for fermions.

\item[(ii)] Crossing equation
\begin{multline}
^{~\text{out,}\bar{1}}\langle\,\theta_{1}\,|\,\mathcal{O}(0)\,|\,\theta
_{2},\dots,\theta_{n}\,\rangle_{2\dots n}^{\text{in,conn.}}\\
=F_{1\ldots n}^{\mathcal{O}}(\theta_{1}+i\pi,\theta_{2},\dots,\theta
_{n})\sigma_{1}^{\mathcal{O}}\mathbf{C}^{\bar{1}1}=F_{2\ldots n1}%
^{\mathcal{O}}(\theta_{2},\dots,\theta_{n},\theta_{1}-i\pi)\mathbf{C}%
^{1\bar{1}}\label{1.12}%
\end{multline}
with the charge conjugation matrix $\mathbf{C}^{\bar{1}1}$ and the statistics
factor $\sigma_{1}^{\mathcal{O}}$ of the operator $\mathcal{O}$ with respect
to the particle $1$.

\item[(iii)] Recursion equation
\begin{equation}
\operatorname*{Res}_{\theta_{12}=i\pi}F_{1\dots n}^{\mathcal{O}}(\theta
_{1},\dots,\theta_{n})=2i\,\mathbf{C}_{12}\,F_{3\dots n}^{\mathcal{O}}%
(\theta_{3},\dots,\theta_{n})\left(  \mathbf{1}-\sigma_{2}^{\mathcal{O}%
}\left(  \sigma S\right)  _{2n}\dots\left(  \sigma S\right)  _{23}\right)
\,,\label{1.14}%
\end{equation}

\item[(iv)] Because there are bound states in the model the function
$F_{\underline{\alpha}}^{\mathcal{O}}({\underline{\theta}})$ has additional
poles. If for instance the particles $1$ and $2$ form a bound state $(12)$,
there is a pole at $\theta_{12}=i\eta,~(0<\eta<\pi)$ such that
\begin{equation}
\operatorname*{Res}_{\theta_{12}=i\eta}F_{12\dots n}^{\mathcal{O}}(\theta
_{1},\theta_{2},\dots,\theta_{n})\,=F_{(12)\dots n}^{\mathcal{O}}%
(\theta_{(12)},\dots,\theta_{n})\,\sqrt{2}\Gamma_{12}^{(12)}\label{1.16}%
\end{equation}
where the bound state intertwiner $\Gamma_{12}^{(12)}$ of (\ref{b1.40}) and
the values of $\theta_{1},~\theta_{2},~\theta_{(12)}$ and $\eta$ are given in
\cite{K1,KT1,BK}.

\item[(v)] Lorentz covariance
\begin{equation}
F_{1\dots n}^{\mathcal{O}}(\theta_{1}+\mu,\dots,\theta_{n}+\mu)=e^{s\mu
}\,F_{1\dots n}^{\mathcal{O}}(\theta_{1},\dots,\theta_{n})\label{1.18}%
\end{equation}
if the local operator transforms under Lorentz transformations as
$\mathcal{O}\rightarrow e^{s\mu}\mathcal{O}$ where $s$ is the
\textquotedblleft spin\textquotedblright\ of $\mathcal{O}$.
\end{itemize}

The statistics factors in (ii)\textbf{ }and (iii) are not arbitrary, but
consistency and crossing implies that both are the same and that the for
anti-particle $\sigma_{1}^{\mathcal{O}}\sigma_{\bar{1}}^{\mathcal{O}}=1$ holds
(see also \cite{BFK1}). In \cite{BFKZ,BK} was shown that the form factor
equations follow from general LSZ assumptions and \textquotedblleft maximal
analyticity\textquotedblright.

\paragraph{Minimal form factors:}

The solutions of Watson's and the crossing equations (i) and (ii) for two
particles
\begin{align}
F\left(  \theta\right)   &  =S\left(  \theta\right)  F\left(  -\theta\right)
\label{watson}\\
F\left(  i\pi-\theta\right)   &  =F\left(  i\pi+\theta\right)  \nonumber
\end{align}
with no poles in the physical strip $0\leq\operatorname{Im}\theta\leq\pi$ and
at most a simple zero at $\theta=0$ are the minimal form factors \cite{KW}%
\begin{align}
F_{+}^{\min}\left(  \theta\right)   &  =\exp\int_{0}^{\infty}\frac{dt}{t\sinh
t}\frac{e^{-t\left(  1-\nu\right)  }-e^{-t}}{1+e^{-t}}\left(  1-\cosh t\left(
1-\frac{\theta}{i\pi}\right)  \right)  \label{F+}\\
F_{-}^{\min}\left(  \theta\right)   &  =\frac{\cosh\frac{1}{2}\left(
i\pi-\theta\right)  \Gamma^{2}\left(  \frac{1}{2}+\frac{1}{2}\nu\right)
}{\Gamma\left(  1+\frac{1}{2}\nu-\frac{1}{2\pi i}\theta\right)  \Gamma\left(
\frac{1}{2}\nu+\frac{1}{2\pi i}\theta\right)  }F_{+}^{\min}\left(
\theta\right)  \label{F-}\\
F_{0}^{\min}\left(  \theta\right)   &  =\frac{2\tanh\frac{1}{2}\left(
i\pi-\theta\right)  }{i\pi-\theta}F_{-}^{\min}\left(  \theta\right)
\,.\label{F0}%
\end{align}
They belong to the S-matrix eigenvalues $S_{\pm}=b\pm c$ and $S_{0}=b+c+Nd$
(see (\ref{S})). For the construction of the off-shell Bethe Ansatz the
minimal solution of the form factor equation (\ref{1.10}) for the highest
weight eigenvalue of the $O(N)$ S-matrix
\begin{equation}
F\left(  \theta\right)  =\sigma S_{+}\left(  \theta\right)  F\left(
-\theta\right)  =-a(\theta)F\left(  -\theta\right)  \label{Fmin0}%
\end{equation}
is essential. We take the solution\footnote{The minus sign in (\ref{Fmin0})
and the factor $\cosh\frac{1}{2}\left(  i\pi-\theta\right)  $ is due to
fermionic statistics of the fundamental particles (see also eq. 4.14 of
\cite{BFKZ}).}%
\begin{align}
F\left(  \theta\right)   &  =c\cosh\frac{1}{2}\left(  i\pi-\theta\right)
\,F_{+}^{\min}\left(  \theta\right)  \label{Fmin}\\
&  =c\exp\left(  \int_{0}^{\infty}\frac{dt}{t\sinh t}\frac{1+e^{-t\left(
1-\nu\right)  }}{1+e^{-t}}\left(  1-\cosh t\left(  1-\frac{\theta}{i\pi
}\right)  \right)  \right)  \nonumber
\end{align}
or\footnote{Private communication: Karol K. Kozlowski pointed out to one of
the authors (M.K.), that the minimal form factors may be expressed in terms of
Barnes G-function.}%
\[
F\left(  \theta\right)  =\frac{G\left(  \frac{1}{2}\frac{\theta}{i\pi}\right)
G\left(  1-\frac{1}{2}\frac{\theta}{i\pi}\right)  }{G\left(  \frac{1}{2}%
+\frac{1}{2}\frac{\theta}{i\pi}\right)  G\left(  \frac{3}{2}-\frac{1}{2}%
\frac{\theta}{i\pi}\right)  }\frac{G\left(  \frac{1}{2}-\frac{1}{2}\nu
+\frac{1}{2}\frac{\theta}{i\pi}\right)  G\left(  \frac{3}{2}-\frac{1}{2}%
\nu-\frac{1}{2}\frac{\theta}{i\pi}\right)  }{G\left(  1-\frac{1}{2}\nu
+\frac{1}{2}\frac{\theta}{i\pi}\right)  G\left(  2-\frac{1}{2}\nu-\frac{1}%
{2}\frac{\theta}{i\pi}\right)  }%
\]
where $G\left(  z\right)  $ is Barnes G-function, which satisfies (see e.g.
\cite{Wo})%
\[
G\left(  1+z\right)  =\Gamma\left(  z\right)  G\left(  z\right)  \,.
\]
For convenience we have introduced the constant $c$ (see (\ref{FF}))
\begin{equation}
c=G^{2}\left(  \tfrac{1}{2}\right)  G^{2}\left(  1-\tfrac{1}{2}\nu\right)
G^{-2}\left(  \tfrac{3}{2}-\tfrac{1}{2}\nu\right)  \,.\label{c}%
\end{equation}
The full 2-particle form factors are%
\begin{equation}
F_{+,-,0}\left(  \theta\right)  =\frac{-\cos^{2}\frac{1}{2}\pi\nu}{\sinh
\frac{1}{2}\left(  \theta-i\pi\nu\right)  \sinh\frac{1}{2}\left(  \theta
+i\pi\nu\right)  }F_{+,-,0}^{\min}\left(  \theta\right)  \,.\label{Ffull}%
\end{equation}
They are non-minimal solutions of (\ref{watson}) containing the bound state
pole at $\theta=i\pi\nu$ (see (2.16) of \cite{KW}).

\section{$O(N)$ form factors and Bethe Ansatz}

\label{s4}

\subsection{The fundamental Theorem}

Following \cite{KW} we write the general form factor $F_{1\dots n}%
^{\mathcal{O}}(\underline{\theta})$ for n-fundamental particles as
(\ref{2.10}) where $F(\theta)$ is the minimal form factor function
(\ref{Fmin}). The K-function $K_{1\dots n}^{\mathcal{O}}(\underline{\theta})$
is determined by the form factor equations (i) - (v). We propose the
K-function in terms of a nested `off-shell' Bethe Ansatz (\ref{BA}) as a
multiple contour integral.

The basic Bethe Ansatz co-vectors in (\ref{PSI}) are defined as (for more
details see \cite{BFK5,BFK7})%
\begin{equation}
\tilde{\Phi}_{\underline{\alpha}}^{\underline{\mathring{\beta}}}%
(\underline{\theta},\underline{z})=\left(  \Pi_{\underline{\beta}}%
^{\underline{\mathring{\beta}}}(\underline{z})\Omega\tilde{T}_{1}^{\beta_{m}%
}(\underline{\theta},z_{m})\dots\tilde{T}_{1}^{\beta_{1}}(\underline{\theta
},z_{1})\right)  _{\underline{\alpha}}\,.\label{BS}%
\end{equation}
The matrix $\Pi_{\underline{\beta}}^{\underline{\mathring{\beta}}}%
(\underline{z})$ intertwines\footnote{This matrix $\Pi$ is trivial for the
$SU(N)$ Bethe Ansatz because the $SU(N)$ S-matrix does not depend on $N$ for a
suitable normalization and parametrization.} between the S-matrices $S$ of
$O(N)$ and $\mathring{S}$ of $O(N-2)$%
\begin{equation}
\tilde{\mathring{S}}_{ij}(z_{ij}\mathring{\nu}/\nu)\Pi_{\dots ij\dots
}(\underline{z})=\Pi_{\dots ji\dots}(\underline{z})\tilde{S}_{ij}%
(z_{ij})\,.\label{PiS}%
\end{equation}
The Bethe Ansatz co-vectors (\ref{BS}) are generalizations of vectors
introduced by Tarasov \cite{Tar} for the Korepin-Izergin model. Below we will
use the following relations for special components of $\Pi$ (for more details
see \cite{BFK5,BFK6,BFK7})%
\begin{equation}
\Pi_{\underline{\beta}}^{\underline{\mathring{\beta}}}=\left\{
\begin{array}
[c]{lll}%
0 & \text{for} & \beta_{1}=1,~\text{or }\beta_{m}=\bar{1}\\
\delta_{\beta_{1}}^{\mathring{\beta}_{1}}\,\Pi_{\beta_{2}\dots\beta_{m}%
}^{\mathring{\beta}_{2}\dots\mathring{\beta}_{m}} & \text{for} & \beta_{1}%
\neq\bar{1}\\[1mm]%
\Pi_{\beta_{1}\dots\beta_{m-1}}^{\mathring{\beta}_{1}\dots\mathring{\beta
}_{m-1}}\,\delta_{\beta_{m}}^{\mathring{\beta}_{m}} & \text{for} & \beta
_{m}\neq1\,.
\end{array}
\right.  \label{Pi}%
\end{equation}

The scalar function $\tilde{h}(\underline{\theta},\underline{z})$ in
(\ref{BA}) depends only on the S-matrix and not on the specific operator
$\mathcal{O}(x)$%
\begin{equation}
\tilde{h}(\underline{\theta},\underline{z})=\prod_{i=1}^{n}\prod_{j=1}%
^{m}\tilde{\phi}(\theta_{i}-z_{j})\prod_{1\leq i<j\leq m}\tau(z_{i}-z_{j})\,.
\label{h}%
\end{equation}

\begin{figure}[h]
$%
\begin{array}
[c]{ccc}%
\fbox{\unitlength2.6mm\begin{picture}(26,15)(0,2) \thicklines \def\ff#1{ \put(-.05,13.2){$_{\bullet~\theta_{#1}+i\pi(4-\nu)}$} \put(-.05,9.2){$_{\bullet~\theta_{#1}+i\pi(2-\nu)}$} \put(.0,2.2){$_{\bullet~\theta_{#1}-2\pi i}$} \put(.25,3.9){${\hbox{\circle{.3}}}~_{\theta_{#1}-i\pi }$} \put(.25,5){${\hbox{\circle{.3}}}~_{\theta_{#1}-i\pi\nu}$} \put(0,6.2){$_{\bullet~\theta_{#1}}$} \put(.25,5){\hbox{\circle{1.15}}} \put(.7,4.5){\vector(1,0){0}} \put(.2,14.5){\oval(.8,12)[b]} \put(-.2,11){\vector(0,-1){0}} } \put(1,0){\ff{n}} \put(8,3){\dots} \put(12,0){\ff{2}} \put(20,2){\ff{1}} \end{picture}~~}%
~ &  &
\fbox{\unitlength2.6mm\begin{picture}(26,15)(0,-3) \thicklines \def\ff#1{ \put(-.27,9.2){$_{\bullet~\theta_{#1}+i\pi(2-\nu)}$} \put(0.2,-1.8){$_{\bullet~\theta_{#1}-4\pi i}$} \put(0.16,2.2){$_{\bullet~\theta_{#1}-2\pi i}$} \put(.28,4){${\hbox{\circle{.3}}}~_{\theta_{#1}-i\pi }$} \put(.28,5.1){${\hbox{\circle{.3}}}~_{\theta_{#1}-i\pi\nu}$} \put(.16,6.2){$_{\bullet~\theta_{#1}}$} \put(.4,6.25){\hbox{\circle{1.15}}} \put(.85,6.8){\vector(1,0){0}} \put(.4,-3){\oval(.8,12)[t]} \put(-.04,0){\vector(0,1){0}} } \put(1,0){\ff{n}} \put(8,3){\dots} \put(12,0){\ff{2}} \put(20,2){\ff{1}} \end{picture}~~}%
\\[1mm]%
\text{~(odd)} &  & \text{(even)}%
\end{array}
$\caption{The integration contours $\mathcal{C}_{\underline{\theta}}^{(o)}$
and $\mathcal{C}_{\underline{\theta}}^{(e)}$. The bullets refer to poles of
the integrand resulting from $\tilde{\phi}(\theta_{i}-z_{j})$ and the small
open circles refer to poles originating from $\tilde{S}(\theta_{i}-z_{j})$.}%
\label{f5.1}%
\end{figure}
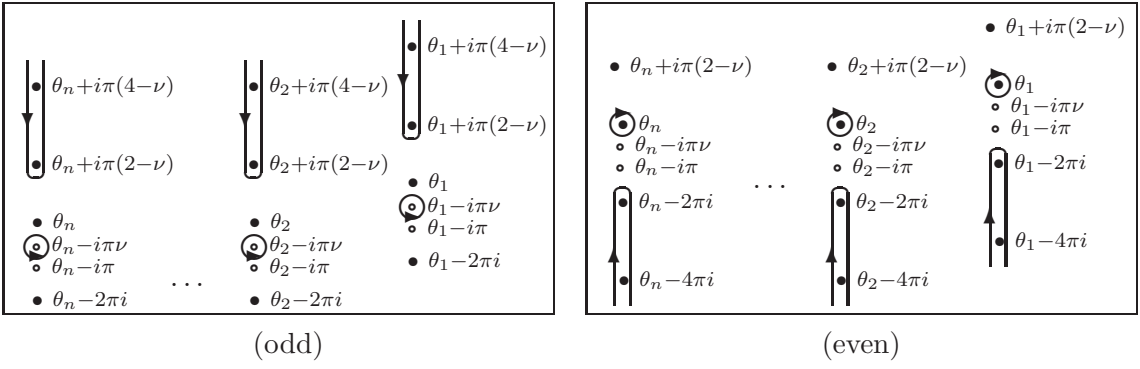The functions $\tilde{\phi}$ and $\tau$ satisfy the shift
equations
\begin{align}
\tilde{\phi}(\theta-2\pi i)  &  =-\tilde{b}(\theta)\tilde{\phi}(\theta
)\label{shiftphi}\\
\tau(z-2\pi i)/\tilde{b}(2\pi i-z)  &  =\tau(z)/\tilde{b}(z) \label{shifttau}%
\end{align}
which are related to the form factor equation (ii) or (\ref{1.12})
\cite{BFK5,BFK6,BFK7}. Here for the $O(N)$ Gross-Neveu model
\begin{equation}
\tau(z)=\frac{1}{\tilde{\phi}(-z)\tilde{\phi}(z)} \label{2.19}%
\end{equation}
where $\,\tilde{\phi}(\theta)$ is\footnote{This is in contrast to the $\sigma
$-model case where the $\tilde{\phi}$-functions depend on whether $j$ in
(\ref{h}) is even or odd.}%
\begin{equation}
\tilde{\phi}(\theta)=\Gamma\left(  1-\frac{1}{2}\nu+\frac{1}{2\pi i}%
\theta\right)  \Gamma\left(  -\frac{1}{2\pi i}\theta\right)  . \label{phi}%
\end{equation}
The form factor equation (iii) or (\ref{1.14}) (as will be discussed in
appendix \ref{sd}) requires that%
\begin{equation}
F(\theta)F(\theta+i\pi)\tilde{\phi}(-\theta-i\pi+i\pi\nu)\tilde{\phi}%
(-\theta)=1. \label{FF}%
\end{equation}
The function (\ref{phi}) satisfies this relation. Notice that the equations
(\ref{phi}) and (\ref{FF}) also determine the normalization constant $c$ in
(\ref{Fmin}) and (\ref{c}).

Similar as in \cite{BFK7} the integration contours $\mathcal{C}_{\underline
{\theta}}^{(j)}$ in (\ref{BA}) depend on whether $j$ is even or odd, they are
depicted in Fig.~\ref{f5.1}. The Ansatz (\ref{2.10}) and (\ref{BA}) transforms
the matrix equations (i) - (v) (see (\ref{1.12})-(\ref{1.18})) into much
simpler scalar equations for the scalar p-function $p^{\mathcal{O}}%
(\underline{\theta},\underline{z})$. This function depends on the specific
operator $\mathcal{O}(x)$ and is in general a simple function of
$e^{\theta_{i}}$ and $e^{z_{j}}.$

\begin{theorem}
\label{TN}Assume that:

\begin{enumerate}
\item The p-function $p^{\mathcal{O}}(\underline{\theta},\underline{z})$
satisfies the equations
\begin{equation}
\left.
\begin{array}
[c]{ll}%
(\mathrm{i}^{\prime}) & p^{\mathcal{O}}(\underline{\theta},\underline
{z})~\text{is symmetric under }\theta_{i}\leftrightarrow\theta_{j}\\[2mm]%
(\mathrm{ii}_{1}^{\prime}) & p^{\mathcal{O}}(\underline{\theta},\underline
{z})=\sigma^{\mathcal{O}}(-1)^{m}p^{\mathcal{O}}(\theta_{1}+2\pi i,\theta
_{2},\dots,\underline{z})\\[1mm]%
(\mathrm{ii}_{2}^{\prime}) & p^{\mathcal{O}}(\underline{\theta},\underline
{z})=(-1)^{n}p^{\mathcal{O}}(\underline{\theta},z_{1}+2\pi i,z_{2},\dots)\\
(\mathrm{iii}^{\prime}) & p^{\mathcal{O}}(\underline{\theta},\underline
{z})=\,p^{\mathcal{O}}(\underline{\check{\theta}},\underline{\check{z}})
\end{array}
\right\}  \label{p}%
\end{equation}
where in $(\mathrm{iii}^{\prime})$ $\theta_{12}=i\pi,~$ $z_{1}=\theta_{1}%
-i\pi\nu$ and $z_{2}=\theta_{2}$. The short notations $\underline
{\check{\theta}}=(\theta_{3},\dots,\theta_{n})$ and $\underline{\check{z}%
}=(z_{3},\dots,z_{m})$ are used.

\item The higher level function $L_{\underline{\beta}}(\underline{z})$ in
(\ref{PSI}) satisfies $(\mathrm{i})^{(k)}$ - $(\mathrm{iii})^{(k)}$ of
(\ref{ik}) -- (\ref{iiik}) for $k=1$.

\item A suitable choice of the normalization constants in (\ref{BA}).\emph{ }
\end{enumerate}

\noindent Then the co-vector valued function $F_{\underline{\alpha}%
}(\underline{\theta})$ given by the Ansatz (\ref{2.10}) and the integral
representation (\ref{BA}) satisfies the form factor equations $(\mathrm{i})$
-- $(\mathrm{v})$ of (\ref{1.10}) -- (\ref{1.18}).
\end{theorem}

\noindent The proof of this\ theorem can be found in appendix \ref{sd}.

\subsection{Higher level off-shell Bethe Ansatz}

For this discussion it is convenient to introduce the variables $u,v$ defined
by $\theta=i\pi\nu_{k}u,~z=i\pi\nu_{k}v$ and $\nu_{k}=2/(N-2k-2)$. For the
$O(N-2k)$ S-matrix $S^{(k)}(u)$ we write as in (\ref{1.1})%
\begin{align}
\tilde{S}^{(k)}(u) &  =S^{(k)}/S_{+}^{(k)}=\tilde{b}(u)\mathbf{1}+\tilde
{c}(u)\mathbf{P}+\tilde{d}_{k}(u)\mathbf{K}\label{Su}\\
\tilde{b}(u) &  =\frac{u}{u-1},~\tilde{c}(u)=\frac{-1}{u-1},~\tilde{d}%
_{k}(u)=\frac{u}{u-1}\frac{1}{u-1/\nu_{k}}\,.\nonumber
\end{align}
and define%
\begin{align}
K_{\underline{\alpha}}^{(k)}(\underline{u}) &  =\tilde{N}_{m_{k}}^{(k)}%
\int_{\mathcal{C}_{\underline{u}}^{(1)}}dv_{1}\cdots\int_{\mathcal{C}%
_{\underline{u}}^{(m_{k})}}dv_{m_{k}}\,\tilde{h}(\underline{u},\underline
{v})p^{(k)}(\underline{u},\underline{v})\,\,\tilde{\Psi}_{\underline{\alpha}%
}^{(k)}(\underline{u},\underline{v})\label{Kk}\\
\tilde{\Psi}_{\underline{\alpha}}^{(k)}(\underline{u},\underline{v}) &
=L_{\underline{\mathring{\beta}}}^{(k)}(\underline{v})\,\left(  \tilde{\Phi
}^{(k)}\right)  _{\underline{\alpha}}^{\underline{\mathring{\beta}}%
}(\underline{u},\underline{v}),\quad L_{\underline{\mathring{\beta}}}%
^{(k)}(\underline{v})=K_{\underline{\mathring{\beta}}}^{(k+1)}(\underline
{v})\nonumber
\end{align}
with $\underline{u}=u_{1},\dots,u_{n_{k}},~\underline{v}=v_{1},\dots,v_{m_{k}%
}$ and $m_{k}=n_{k+1}$.

The equations (i)$^{(k)}$-(iii)$^{(k)}$ for $k>0$ are in terms of these
variables similar as in \cite{BFK7}

\begin{itemize}
\item[(i)$^{(k)}$]
\begin{equation}
K_{\dots ij\dots}^{(k)}(\dots,u_{i},u_{j},\dots)=K_{\dots ji\dots}^{(k)}%
(\dots,u_{j},u_{i},\dots)\,\tilde{S}_{ij}^{(k)}(u_{ij}) \label{ik}%
\end{equation}

\item[(ii)$^{(k)}$]
\begin{equation}
K_{1\ldots n_{k}}^{(k)}(u_{1}+2/\nu,u_{2},\dots,u_{n_{k}})\mathbf{C}^{\bar
{1}1}=K_{2\ldots n_{k}1}^{(k)}(u_{2},\dots,u_{n_{k}},u_{1})\mathbf{C}%
^{1\bar{1}} \label{iik}%
\end{equation}

\item[(iii)$^{(k)}$]
\begin{equation}
\operatorname*{Res}_{u_{12}=1/\nu_{k}}K_{1\dots n_{k}}^{(k)}(u_{1}%
,\dots,u_{n_{k}})=\prod_{i=3}^{n_{k}}\tilde{\phi}(u_{i1}+1)\tilde{\phi}%
(u_{i2})\mathbf{C}_{12}K_{3\dots n_{k}}^{(k)}(u_{3},\dots,u_{n_{k}%
})\,.\label{iiik}%
\end{equation}
in addition we have here the bound state relation

\item[(iv)$^{(k)}$]
\begin{equation}
\operatorname*{Res}_{u_{12}=1}F_{12\dots n}^{(k)}(u_{1},u_{2},\dots
,u_{n})\,=F_{(12)\dots n}^{(k)}(u_{(12)},\dots,u_{n})\,\sqrt{2}\Gamma
_{12}^{(12)}. \label{ivk}%
\end{equation}

\end{itemize}

The form factor equations (i) - (iv) of (\ref{1.10})-(\ref{1.16}) for
$O(N-2k)$ are similar to these higher level equations. There are, however, two
differences: 1) The shift in (ii)$^{(k)}$ is the one of $O(N)$ but not that of
$O(N-2k)$. 2) There is only one term on the right hand side in (iii)$^{(k)}$.

We assume that the p-function $p^{(k)}(\underline{u},\underline{v})$ satisfies
the equations
\begin{equation}%
\begin{array}
[c]{ll}%
(\mathrm{i}^{\prime}) & p^{(k)}(\underline{u},\underline{v})~\text{is
symmetric under }u_{i}\leftrightarrow u_{j},~v_{i}\leftrightarrow v_{j}\\[2mm]%
(\mathrm{ii}^{\prime}) & p^{(k)}(\underline{u},\underline{v})=(-1)^{m_{k}%
}p^{(k)}(u_{1}+2/\nu,u_{2},\dots,\underline{v})=(-1)^{n_{k}}p^{(k)}%
(\underline{u},v_{1}+2/\nu,v_{2},\dots)\\
(\mathrm{iii}^{\prime}) & p^{(k)}(\underline{u},\underline{z})=\,p^{(k)}%
(\underline{\check{u}},\underline{\check{v}})~\text{for }u_{12}=1/\nu
_{k},~v_{1}=u_{1}-1~\text{and }v_{2}=u_{2}\,
\end{array}
\label{pk}%
\end{equation}
where we use the short notations $\underline{\check{u}}=(u_{3},\dots,u_{n_{k}%
})$ and $\underline{\check{v}}=(v_{3},\dots,v_{m_{k}})$.

\begin{lemma}
\label{L1}For $0<k<\frac{1}{2}\left(  N-4\right)  $ the functions
$K_{\underline{\alpha}}^{(k)}(\underline{u})$ of (\ref{Kk}) satisfy the
equations $(\mathrm{i})^{(k)},~(\mathrm{ii})^{(k)}$ and $(\mathrm{iii})^{(k)}%
$, if the corresponding relations are satisfied for $K_{\underline
{\mathring{\beta}}}^{(k+1)}(\underline{v})$ and if suitable choice of the
normalization constants in (\ref{Kk}) is assumed. The weights of the operator
$\mathcal{O}$%
\begin{equation}
w^{\mathcal{O}}=(w_{1},\dots,w_{N/2})=\left(  n_{0}-n_{1},\dots,n_{N/2-2}%
-n_{-}-n_{+},n_{-}-n_{+}\right)  \, \label{w}%
\end{equation}
determine the numbers $m_{k}=n_{k+1}$ for a given number of particles
$n=n_{0}$
\end{lemma}

\noindent The proof of this lemma can be found in appendix \ref{se}.

\section{Examples}

\label{s5}

In this section, to illustrate our general results we present some simple examples.

\subsection{Current}

The $O(N)$ Noether current\footnote{In the real basis.}%
\[
J_{\mu}^{\alpha\beta}=\bar{\psi}^{\alpha}\gamma_{\mu}\psi^{\beta}%
\]
transforms as the antisymmetric tensor representation of $O(N)$. This operator
has therefore the weights $w^{J}=(w_{1},\dots,w_{N/2})=(1,1,0,\dots,0)$ (see
\cite{BFK5,BFK6}), which implies with (\ref{w}) that%
\[
n-2=n_{1}-1=n_{2}=\dots=n_{N/2-2}=n_{-}+n_{+},n_{-}=n_{+}\,.
\]
where $n_{i}$ are the numbers of integrations in the various levels of the
off-shell Bethe Ansatz. The existence of a pseudo-potential $J^{\alpha\beta
}(x)$ follows from the conservation law $\partial^{\mu}J_{\mu}^{\alpha\beta
}=0$
\[
J_{\mu}^{\alpha\beta}(x)=\epsilon_{\mu\nu}\partial^{\nu}J^{\alpha\beta}(x).
\]
For the form factors of both operators we have the relation%
\begin{equation}
F_{\underline{\alpha}}^{J_{\mu}}(\underline{\theta})=-i\epsilon_{\mu\nu
}\left(  {{\sum\,}}p_{i}^{\nu}\right)  F_{\underline{\alpha}}^{J}%
(\underline{\theta})\,. \label{J}%
\end{equation}
Because the Bethe Ansatz yields highest weight states we obtain the matrix
elements of the highest weight component of $J^{\alpha\beta}$ which means in
the complex basis $J(x)=J^{12}(x).$

We propose the form factors of the operator $J(x)$ (for $n=m+1=n_{1}%
+1=n_{2}+2$ even)%
\begin{align*}
\langle0|J(0)|\underline{\theta}\rangle_{\underline{\alpha}} &  =F_{\underline
{\alpha}}^{J}(\underline{\theta})=\prod_{i<j}F(\theta_{ij})K_{\underline
{\alpha}}^{J}(\underline{\theta})\\
K_{\underline{\alpha}}^{J}(\underline{\theta}) &  =N_{n}^{J}\int
_{\mathcal{C}_{\underline{\theta}}^{(1)}}dz_{1}\dots\int_{\mathcal{C}%
_{\underline{\theta}}^{(m)}}dz_{m}\,\tilde{h}(\underline{\theta},\underline
{z})p^{J}(\underline{\theta},\underline{\underline{z}})\,\tilde{\Psi
}_{\underline{\alpha}}(\underline{\theta},\underline{z})\,.
\end{align*}
Expressing $\tilde{\Psi}_{\underline{\alpha}}(\underline{\theta},\underline
{z})$ in terms of all higher level Bethe Ansatzes there appears the product of
all level p-functions $p^{J}(\underline{\theta},\underline{\underline{z}})$.
For the example of the current it depends on the $\theta_{i}$ and the second
level $z_{j}^{(2)}$
\begin{equation}
p^{J}(\underline{\theta},\underline{\underline{z}})=e^{\frac{1}{2}%
\Big(\sum\limits_{i=1}^{n}\theta_{i}-\sum\limits_{j=1}^{n_{2}}z_{j}%
^{(2)}-\frac{1}{2}n_{2}i\pi\mathring{\nu}\Big)}/\sum\limits_{i=1}^{n}%
e^{\theta_{i}}+e^{-\frac{1}{2}\Big(\sum\limits_{i=1}^{n}\theta_{i}%
-\sum\limits_{j=1}^{n_{2}}z_{j}^{(2)}-\frac{1}{2}n_{2}i\pi\mathring{\nu}%
\Big)}/\sum\limits_{i=1}^{n}e^{-\theta_{i}}\label{pJ}%
\end{equation}
which satisfies (\ref{p}) with
\[%
\begin{array}
[c]{lcl}%
\text{charge} &  & Q^{J}=0\\
\text{weight vector} &  & w^{J}=\left(  1,1,0,\dots,0\right)  \\
\text{statistics factor} &  & \sigma^{J}=1\\
\text{spin} &  & s^{J}=0,~s^{J_{\mu}}=1\,.
\end{array}
\]
For example for 2-particle form factor we obtain (see appendix \ref{sh})%
\begin{align}
F_{\alpha_{1}\alpha_{2}}^{J^{\alpha\beta}}\left(  \theta_{1},\theta
_{2}\right)   &  =im\left(  \delta_{\alpha_{1}}^{\alpha}\delta_{\alpha_{2}%
}^{\beta}-\delta_{\alpha_{1}}^{\beta}\delta_{\alpha_{2}}^{\alpha}\right)
\frac{1}{\cosh\frac{1}{2}\theta_{12}}F_{-}\left(  \theta\right)  \label{FJ}\\
F_{\alpha_{1}\alpha_{2}}^{J_{\mu}^{\alpha\beta}}\left(  \theta_{1},\theta
_{2}\right)   &  =i\left(  \delta_{\alpha_{1}}^{\alpha}\delta_{\alpha_{2}%
}^{\beta}-\delta_{\alpha_{1}}^{\beta}\delta_{\alpha_{2}}^{\alpha}\right)
\bar{v}(\theta_{1})\gamma_{\mu}u(\theta_{2})F_{-}\left(  \theta\right)
\label{FJmu}%
\end{align}
with $F_{-}\left(  \theta\right)  $ of (\ref{Ffull}) and (\ref{F-}), $\bar
{v}(\theta_{1})\gamma^{\pm}u(\theta_{2})=\pm i2me^{\pm\frac{1}{2}(\theta
_{1}+\theta_{2})}$ and (\ref{J}). This result agrees with \cite{KW}.

\subsection{Field}

For the fundamental field $\psi^{\alpha}(x)$ in (\ref{LGN}) the numbers
$n_{i}$ of integrations in the various levels of the off-shell Bethe Ansatz
satisfy%
\[
n-1=n_{1}=n_{2}=\dots=n_{N/2-2}=n_{-}+n_{+},n_{-}=n_{+}\,
\]
because $\psi^{\alpha}$ transforms as the vector representation of $O(N)$ (see
\cite{BFK5,BFK6} and (\ref{w})). We restrict to component $\psi=\psi^{1}$
because as usual the Bethe Ansatz yields highest weight states. For
convenience we multiply the field with the Dirac operator and take%
\begin{equation}
\chi{(}x{)}=i(-i\gamma\partial+m)\psi(x)\,. \label{chi}%
\end{equation}
We propose for the $n$-particle form factors ($n=m+1$ odd) for the spinor
components $\chi^{(\pm)}$%
\begin{align}
\langle0|\chi^{(\pm)}(0)|\underline{\theta}\rangle_{\underline{\alpha}}  &
=F_{\underline{\alpha}}^{\chi^{(\pm)}}(\underline{\theta})=\prod_{i<j}%
F(\theta_{ij})K_{\underline{\alpha}}^{\chi^{(\pm)}}(\underline{\theta
})\label{psi0}\\
K_{\underline{\alpha}}^{\chi^{(\pm)}}(\underline{\theta})  &  =N_{n}^{\chi
}\int_{\mathcal{C}_{\underline{\theta}}^{(1)}}dz_{1}\dots\int_{\mathcal{C}%
_{\underline{\theta}}^{(m)}}dz_{m}\,\tilde{h}(\underline{\theta},\underline
{z})p^{\chi^{(\pm)}}(\underline{\theta},\underline{z})\,\tilde{\Psi
}_{\underline{\alpha}}(\underline{\theta},\underline{z}) \label{psi1}%
\end{align}
with the p-function (for $n=m+1=$ odd $>1$)%
\begin{equation}
p^{\chi^{(\pm)}}(\underline{\theta},\underline{z})=\exp\left(  \mp\frac{1}%
{2}\left(  {{\textstyle\sum\limits_{j=1}^{n}}}\theta_{j}-{\textstyle\sum
\limits_{j=1}^{m}}z_{j}-\frac{1}{2}mi\pi\nu\right)  \right)  \label{pchi}%
\end{equation}
which solves (\ref{p}) with
\[%
\begin{array}
[c]{lcl}%
\text{charge} &  & Q^{\psi}=1\\
\text{weight vector} &  & w^{\psi}=\left(  1,\dots,0\right) \\
\text{statistics factor} &  & \sigma^{\psi}=-1\\
\text{spin} &  & s^{\psi}=\frac{1}{2}%
\end{array}
\]
The one particle form factor is trivial%
\[
\langle0|\psi(0)|\theta\rangle_{\alpha}=F_{\alpha}^{\psi}(\theta
)=\delta_{_{\alpha}}^{1}\,u(\theta)\,.
\]
For the three particle form factor ($n=3,~m=2$) the equations (\ref{psi0}) and
(\ref{psi1}) write as%
\begin{align}
\langle0|\chi(0)|\underline{\theta}\rangle_{\underline{\alpha}}  &
=F_{\underline{\alpha}}^{\chi}(\underline{\theta})=F(\theta_{12})F(\theta
_{13})F(\theta_{23})K_{\underline{\alpha}}^{\chi}(\underline{\theta
})\nonumber\\
K_{\underline{\alpha}}^{\chi}(\underline{\theta})  &  =N_{3}^{\chi}%
\int_{\mathcal{C}_{\underline{\theta}}^{(o)}}dz_{1}\int_{\mathcal{C}%
_{\underline{\theta}}^{(e)}}dz_{2}\,\tilde{h}(\underline{\theta},\underline
{z})p^{\chi}(\underline{\theta},\underline{z})\,\tilde{\Psi}_{\underline
{\alpha}}(\underline{\theta},\underline{z}) \label{phi3'}%
\end{align}
with%
\begin{align*}
\tilde{h}(\underline{\theta},\underline{z})  &  =\prod_{i=1}^{3}\tilde{\phi
}(\theta_{i}-z_{1})\tilde{\phi}(\theta_{i}-z_{2})\frac{1}{\tilde{\phi}%
(z_{12})\tilde{\phi}(-z_{12})}\\
p^{\chi^{\pm}}(\underline{\theta},\underline{z})  &  =e^{\mp\frac{1}{2}\left(
\theta_{1}+\theta_{2}+\theta_{3}-z_{1}-z_{2}-i\pi\nu\right)  }\\
\tilde{\Psi}_{\underline{\alpha}}(\underline{\theta},\underline{z})  &
=L_{\underline{\mathring{\beta}}}(\underline{z})\,\left(  \Pi_{\underline
{\beta}}^{\underline{\mathring{\beta}}}(\underline{z})\Omega\tilde{T}%
_{1}^{\beta_{2}}(\underline{\theta},z_{2})\tilde{T}_{1}^{\beta_{1}}%
(\underline{\theta},z_{1})\right)  _{\underline{\alpha}}.
\end{align*}
Lemma \ref{L1} for the $O(N-2)$ weights $w=(0,\dots,0)$ yields for the higher
level function $L_{\mathring{\beta}_{1}\mathring{\beta}_{2}}(\underline
{z})=\mathbf{\mathring{C}}_{\mathring{\beta}_{1}\mathring{\beta}_{2}}%
L(z_{12})$ with%
\begin{equation}
L(z)=\frac{\Gamma\left(  1-\frac{1}{2}\nu-\frac{z}{2\pi i}\right)
\Gamma\left(  -\frac{1}{2}\nu+\frac{z}{2\pi i}\right)  }{\Gamma\left(
1+\frac{1}{2}\left(  1-\nu\right)  -\frac{z}{2\pi i}\right)  \Gamma\left(
\frac{1}{2}\left(  1-\nu\right)  +\frac{z}{2\pi i}\right)  } \label{Lem}%
\end{equation}
(see appendix \ref{sf}). We could not perform the integrations\footnote{Doing
one integral we obtain a generalization of Meijer's G-functions. The second
integration does not yield known functions (to our knowledge). One could, of
course, apply numerical integration techniques and determine the asymptotic
behavior for large $\theta^{\prime}$s which is under investigation
\cite{BFK9}.} in (\ref{phi3'}) for general $N$, but we expand the exact
expression in $1/N$-expansion to compare the result with the $1/N$-expansion
of the $O(N)$ Gross-Neveu model in terms of Feynman graphs.

\paragraph{1/N expansion:}

We obtain the 3-particle form factor of $\chi(x)$ up to $O(N^{-2})$ as (see
appendix \ref{sg})%
\begin{equation}
{F_{\alpha\beta\gamma}^{\chi^{\delta}}}=\frac{8\pi m}{N}\,\left(
\delta_{\gamma}^{\delta}\mathbf{C}_{\alpha\beta}\frac{\cosh\frac{1}{2}%
\theta_{12}}{\theta_{12}-i\pi}\,u(\theta_{3})-\delta_{\beta}^{\delta
}\mathbf{C}_{\alpha\gamma}\frac{\cosh\frac{1}{2}\theta_{13}}{\theta_{13}-i\pi
}\,u(\theta_{2})+\delta_{\alpha}^{\delta}\mathbf{C}_{\beta\gamma}\frac
{\cosh\frac{1}{2}\theta_{23}}{\theta_{23}-i\pi}\,u(\theta_{1})\right)
\label{F3g}%
\end{equation}
which agrees with the $1/N$ expansion using Feynman graphs (see appendix
\ref{s1overN}).

\paragraph{Bound state form factor of $\psi$:}

We discuss the bound state fusion of 2 fundamental fermions $f+f\rightarrow
b_{2}$, a boson of mass $m_{2}$ (see (\ref{mk})). Writing (\ref{chi}) as
\[
\psi(x)=(i\gamma\partial+m)\tilde{\chi}{(}x{),~~}\tilde{\chi}{(}x{)=-i}\left(
\square+m^{2}\right)  ^{-1}\chi{(}x{)\,,}%
\]
we apply the form factor equation (iv), i.e. (\ref{1.16})%
\[
\operatorname*{Res}_{\theta_{12}=i\pi\nu}F_{123}^{\mathcal{O}}(\theta
_{1},\theta_{2},\theta_{3})\,=F_{(12)3}^{\mathcal{O}}(\theta_{(12)},\theta
_{3})\,\sqrt{2}\Gamma_{12}^{(12)}%
\]
to the operator\footnote{Strictly speaking
$F_{1\bar{1}1}^{\tilde{\chi}}\pm F_{\bar{1}11}^{\tilde{\chi}}$ give
$F_{b_{2}^{(0,2)}1}^{\tilde{\chi}}$.} $\mathcal{O}=\tilde{\chi}$
\[
\operatorname*{Res}\limits_{\theta_{12}=i\pi\nu}F_{1\bar{1}1}^{\tilde{\chi
}^{(\pm)}}(\underline{\theta})=F_{b_{2}1}^{\tilde{\chi}^{(\pm)}}(\theta
_{0},\theta_{3})\sqrt{2}\Gamma_{1\bar{1}}^{b_{2}}\,.
\]
The result may be written as
\begin{multline*}
F_{b_{2}1}^{\tilde{\chi}^{(\pm)}}(\theta_{0},\theta_{3})=e^{\mp\frac{1}%
{2}\theta_{0}}\left(  e^{\pm\frac{1}{4}i\pi\nu}f_{13}(\theta_{03})+e^{\mp
\frac{1}{4}i\pi\nu}f_{32}(\theta_{03})\right) \\
+e^{\mp\frac{1}{2}\theta_{3}}\left(  e^{\pm\frac{1}{2}i\pi\nu}f_{11}%
(\theta_{03})+e^{\mp\frac{1}{2}i\pi\nu}f_{22}(\theta_{03})\right)  .
\end{multline*}
where the functions $f_{ij}$ may be calculated in terms of hypergeometric
functions $_{3}F_{2}$ (for more details see appendix \ref{siv}). For example
$f_{13}$ is plotted for $N=12$ in Fig. \ref{f13a}.

\begin{figure}[h]
\begin{center}
\setlength{\unitlength}{0.240900pt}
\ifx\plotpoint\undefined\newsavebox{\plotpoint}\fi
\begin{picture}(1049,629)(0,0)
\font\gnuplot=cmr10 at 10pt
\gnuplot
\sbox{\plotpoint}{\rule[-0.200pt]{0.400pt}{0.400pt}}%
\put(60.0,82.0){\rule[-0.200pt]{0.400pt}{4.818pt}}
\put(60,41){\makebox(0,0){ 0}}
\put(60.0,570.0){\rule[-0.200pt]{0.400pt}{4.818pt}}
\put(153.0,82.0){\rule[-0.200pt]{0.400pt}{4.818pt}}
\put(153,41){\makebox(0,0){ 0.1}}
\put(153.0,570.0){\rule[-0.200pt]{0.400pt}{4.818pt}}
\put(246.0,82.0){\rule[-0.200pt]{0.400pt}{4.818pt}}
\put(246,41){\makebox(0,0){ 0.2}}
\put(246.0,570.0){\rule[-0.200pt]{0.400pt}{4.818pt}}
\put(339.0,82.0){\rule[-0.200pt]{0.400pt}{4.818pt}}
\put(339,41){\makebox(0,0){ 0.3}}
\put(339.0,570.0){\rule[-0.200pt]{0.400pt}{4.818pt}}
\put(432.0,82.0){\rule[-0.200pt]{0.400pt}{4.818pt}}
\put(432,41){\makebox(0,0){ 0.4}}
\put(432.0,570.0){\rule[-0.200pt]{0.400pt}{4.818pt}}
\put(525.0,82.0){\rule[-0.200pt]{0.400pt}{4.818pt}}
\put(525,41){\makebox(0,0){ 0.5}}
\put(525.0,570.0){\rule[-0.200pt]{0.400pt}{4.818pt}}
\put(617.0,82.0){\rule[-0.200pt]{0.400pt}{4.818pt}}
\put(617,41){\makebox(0,0){ 0.6}}
\put(617.0,570.0){\rule[-0.200pt]{0.400pt}{4.818pt}}
\put(710.0,82.0){\rule[-0.200pt]{0.400pt}{4.818pt}}
\put(710,41){\makebox(0,0){ 0.7}}
\put(710.0,570.0){\rule[-0.200pt]{0.400pt}{4.818pt}}
\put(803.0,82.0){\rule[-0.200pt]{0.400pt}{4.818pt}}
\put(803,41){\makebox(0,0){ 0.8}}
\put(803.0,570.0){\rule[-0.200pt]{0.400pt}{4.818pt}}
\put(896.0,82.0){\rule[-0.200pt]{0.400pt}{4.818pt}}
\put(896,41){\makebox(0,0){ 0.9}}
\put(896.0,570.0){\rule[-0.200pt]{0.400pt}{4.818pt}}
\put(989.0,82.0){\rule[-0.200pt]{0.400pt}{4.818pt}}
\put(989,41){\makebox(0,0){ 1}}
\put(989.0,570.0){\rule[-0.200pt]{0.400pt}{4.818pt}}
\put(60.0,336.0){\rule[-0.200pt]{223.796pt}{0.400pt}}
\put(60.0,82.0){\rule[-0.200pt]{223.796pt}{0.400pt}}
\put(989.0,82.0){\rule[-0.200pt]{0.400pt}{122.377pt}}
\put(60.0,590.0){\rule[-0.200pt]{223.796pt}{0.400pt}}
\put(60.0,82.0){\rule[-0.200pt]{0.400pt}{122.377pt}}
\put(597,527){\makebox(0,0)[r]{$f_{13}(x)$}}
\put(617.0,527.0){\rule[-0.200pt]{24.090pt}{0.400pt}}
\put(69,335){\usebox{\plotpoint}}
\put(125,333.67){\rule{4.577pt}{0.400pt}}
\multiput(125.00,334.17)(9.500,-1.000){2}{\rule{2.289pt}{0.400pt}}
\put(69.0,335.0){\rule[-0.200pt]{13.490pt}{0.400pt}}
\put(181,332.67){\rule{4.336pt}{0.400pt}}
\multiput(181.00,333.17)(9.000,-1.000){2}{\rule{2.168pt}{0.400pt}}
\put(144.0,334.0){\rule[-0.200pt]{8.913pt}{0.400pt}}
\put(218,331.67){\rule{4.577pt}{0.400pt}}
\multiput(218.00,332.17)(9.500,-1.000){2}{\rule{2.289pt}{0.400pt}}
\put(237,330.67){\rule{4.336pt}{0.400pt}}
\multiput(237.00,331.17)(9.000,-1.000){2}{\rule{2.168pt}{0.400pt}}
\put(255,329.17){\rule{3.900pt}{0.400pt}}
\multiput(255.00,330.17)(10.905,-2.000){2}{\rule{1.950pt}{0.400pt}}
\multiput(274.00,327.95)(3.811,-0.447){3}{\rule{2.500pt}{0.108pt}}
\multiput(274.00,328.17)(12.811,-3.000){2}{\rule{1.250pt}{0.400pt}}
\multiput(292.00,324.93)(1.220,-0.488){13}{\rule{1.050pt}{0.117pt}}
\multiput(292.00,325.17)(16.821,-8.000){2}{\rule{0.525pt}{0.400pt}}
\multiput(311.58,315.21)(0.494,-0.717){25}{\rule{0.119pt}{0.671pt}}
\multiput(310.17,316.61)(14.000,-18.606){2}{\rule{0.400pt}{0.336pt}}
\put(325.17,292){\rule{0.400pt}{1.300pt}}
\multiput(324.17,295.30)(2.000,-3.302){2}{\rule{0.400pt}{0.650pt}}
\put(326.67,284){\rule{0.400pt}{1.927pt}}
\multiput(326.17,288.00)(1.000,-4.000){2}{\rule{0.400pt}{0.964pt}}
\put(328.17,272){\rule{0.400pt}{2.500pt}}
\multiput(327.17,278.81)(2.000,-6.811){2}{\rule{0.400pt}{1.250pt}}
\put(330.17,253){\rule{0.400pt}{3.900pt}}
\multiput(329.17,263.91)(2.000,-10.905){2}{\rule{0.400pt}{1.950pt}}
\put(332.17,219){\rule{0.400pt}{6.900pt}}
\multiput(331.17,238.68)(2.000,-19.679){2}{\rule{0.400pt}{3.450pt}}
\put(334.17,140){\rule{0.400pt}{15.900pt}}
\multiput(333.17,186.00)(2.000,-45.999){2}{\rule{0.400pt}{7.950pt}}
\put(336.17,82){\rule{0.400pt}{11.700pt}}
\multiput(335.17,115.72)(2.000,-33.716){2}{\rule{0.400pt}{5.850pt}}
\put(338.17,82){\rule{0.400pt}{101.700pt}}
\multiput(337.17,82.00)(2.000,296.917){2}{\rule{0.400pt}{50.850pt}}
\put(339.67,536){\rule{0.400pt}{13.009pt}}
\multiput(339.17,563.00)(1.000,-27.000){2}{\rule{0.400pt}{6.504pt}}
\put(341.17,457){\rule{0.400pt}{15.900pt}}
\multiput(340.17,503.00)(2.000,-45.999){2}{\rule{0.400pt}{7.950pt}}
\put(343.17,423){\rule{0.400pt}{6.900pt}}
\multiput(342.17,442.68)(2.000,-19.679){2}{\rule{0.400pt}{3.450pt}}
\put(345.17,404){\rule{0.400pt}{3.900pt}}
\multiput(344.17,414.91)(2.000,-10.905){2}{\rule{0.400pt}{1.950pt}}
\put(347.17,392){\rule{0.400pt}{2.500pt}}
\multiput(346.17,398.81)(2.000,-6.811){2}{\rule{0.400pt}{1.250pt}}
\put(349.17,384){\rule{0.400pt}{1.700pt}}
\multiput(348.17,388.47)(2.000,-4.472){2}{\rule{0.400pt}{0.850pt}}
\put(351.17,378){\rule{0.400pt}{1.300pt}}
\multiput(350.17,381.30)(2.000,-3.302){2}{\rule{0.400pt}{0.650pt}}
\multiput(353.58,375.21)(0.494,-0.717){25}{\rule{0.119pt}{0.671pt}}
\multiput(352.17,376.61)(14.000,-18.606){2}{\rule{0.400pt}{0.336pt}}
\multiput(367.00,356.93)(1.154,-0.488){13}{\rule{1.000pt}{0.117pt}}
\multiput(367.00,357.17)(15.924,-8.000){2}{\rule{0.500pt}{0.400pt}}
\multiput(385.00,348.95)(4.034,-0.447){3}{\rule{2.633pt}{0.108pt}}
\multiput(385.00,349.17)(13.534,-3.000){2}{\rule{1.317pt}{0.400pt}}
\put(404,345.17){\rule{3.700pt}{0.400pt}}
\multiput(404.00,346.17)(10.320,-2.000){2}{\rule{1.850pt}{0.400pt}}
\put(422,343.67){\rule{4.577pt}{0.400pt}}
\multiput(422.00,344.17)(9.500,-1.000){2}{\rule{2.289pt}{0.400pt}}
\put(441,342.67){\rule{4.336pt}{0.400pt}}
\multiput(441.00,343.17)(9.000,-1.000){2}{\rule{2.168pt}{0.400pt}}
\put(199.0,333.0){\rule[-0.200pt]{4.577pt}{0.400pt}}
\put(608,342.67){\rule{4.577pt}{0.400pt}}
\multiput(608.00,342.17)(9.500,1.000){2}{\rule{2.289pt}{0.400pt}}
\put(459.0,343.0){\rule[-0.200pt]{35.894pt}{0.400pt}}
\put(645,343.67){\rule{4.577pt}{0.400pt}}
\multiput(645.00,343.17)(9.500,1.000){2}{\rule{2.289pt}{0.400pt}}
\put(627.0,344.0){\rule[-0.200pt]{4.336pt}{0.400pt}}
\put(682,344.67){\rule{4.577pt}{0.400pt}}
\multiput(682.00,344.17)(9.500,1.000){2}{\rule{2.289pt}{0.400pt}}
\put(701,346.17){\rule{3.900pt}{0.400pt}}
\multiput(701.00,345.17)(10.905,2.000){2}{\rule{1.950pt}{0.400pt}}
\put(720,347.67){\rule{4.336pt}{0.400pt}}
\multiput(720.00,347.17)(9.000,1.000){2}{\rule{2.168pt}{0.400pt}}
\put(738,349.17){\rule{3.900pt}{0.400pt}}
\multiput(738.00,348.17)(10.905,2.000){2}{\rule{1.950pt}{0.400pt}}
\multiput(757.00,351.61)(3.811,0.447){3}{\rule{2.500pt}{0.108pt}}
\multiput(757.00,350.17)(12.811,3.000){2}{\rule{1.250pt}{0.400pt}}
\multiput(775.00,354.60)(2.674,0.468){5}{\rule{2.000pt}{0.113pt}}
\multiput(775.00,353.17)(14.849,4.000){2}{\rule{1.000pt}{0.400pt}}
\multiput(794.00,358.59)(1.332,0.485){11}{\rule{1.129pt}{0.117pt}}
\multiput(794.00,357.17)(15.658,7.000){2}{\rule{0.564pt}{0.400pt}}
\multiput(812.00,365.58)(0.964,0.491){17}{\rule{0.860pt}{0.118pt}}
\multiput(812.00,364.17)(17.215,10.000){2}{\rule{0.430pt}{0.400pt}}
\multiput(831.00,375.58)(0.526,0.495){33}{\rule{0.522pt}{0.119pt}}
\multiput(831.00,374.17)(17.916,18.000){2}{\rule{0.261pt}{0.400pt}}
\multiput(850.58,393.00)(0.495,1.235){33}{\rule{0.119pt}{1.078pt}}
\multiput(849.17,393.00)(18.000,41.763){2}{\rule{0.400pt}{0.539pt}}
\multiput(868.58,437.00)(0.495,4.095){35}{\rule{0.119pt}{3.321pt}}
\multiput(867.17,437.00)(19.000,146.107){2}{\rule{0.400pt}{1.661pt}}
\put(664.0,345.0){\rule[-0.200pt]{4.336pt}{0.400pt}}
\put(895.17,82){\rule{0.400pt}{101.700pt}}
\multiput(894.17,378.92)(2.000,-296.917){2}{\rule{0.400pt}{50.850pt}}
\put(887.0,590.0){\rule[-0.200pt]{1.927pt}{0.400pt}}
\multiput(905.58,82.00)(0.495,3.316){35}{\rule{0.119pt}{2.711pt}}
\multiput(904.17,82.00)(19.000,118.374){2}{\rule{0.400pt}{1.355pt}}
\multiput(924.58,206.00)(0.495,1.142){35}{\rule{0.119pt}{1.005pt}}
\multiput(923.17,206.00)(19.000,40.914){2}{\rule{0.400pt}{0.503pt}}
\multiput(943.00,249.58)(0.561,0.494){29}{\rule{0.550pt}{0.119pt}}
\multiput(943.00,248.17)(16.858,16.000){2}{\rule{0.275pt}{0.400pt}}
\multiput(961.00,265.59)(1.408,0.485){11}{\rule{1.186pt}{0.117pt}}
\multiput(961.00,264.17)(16.539,7.000){2}{\rule{0.593pt}{0.400pt}}
\put(897.0,82.0){\rule[-0.200pt]{1.927pt}{0.400pt}}
\put(60.0,82.0){\rule[-0.200pt]{223.796pt}{0.400pt}}
\put(989.0,82.0){\rule[-0.200pt]{0.400pt}{122.377pt}}
\put(60.0,590.0){\rule[-0.200pt]{223.796pt}{0.400pt}}
\put(60.0,82.0){\rule[-0.200pt]{0.400pt}{122.377pt}}
\end{picture}
\end{center}
\caption{Plot of the bound state form factor function $f_{13}(x),$
$(\theta=i\pi x)$ for $N=12~(\nu=1/5)$.}%
\label{f13a}%
\end{figure}
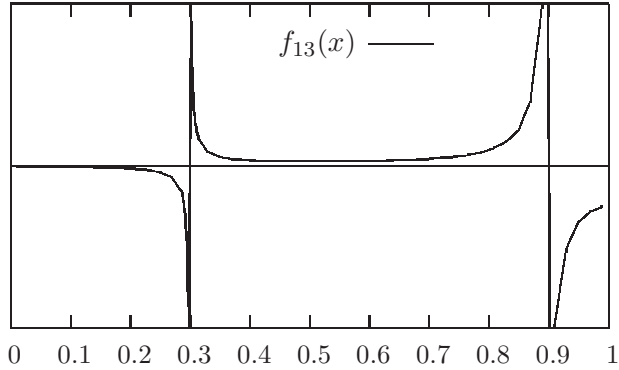The pole at $\theta=\frac{3}{2}i\pi\nu$ (here $x=0.3$) belongs to
the bound state fusion $b_{2}^{(r)}+f\rightarrow b_{3}^{(r\pm1)}$, a fermion
of mass $m_{3}$ (see (\ref{mk})). The pole at $\theta=i\pi\left(  1-\frac
{1}{2}\nu\right)  $ (here $x=0.9$) belongs to the bound state fusion
$b_{2}^{(r)}+f\rightarrow f$, which is again the fundamental fermion. These
are examples of the general \textquotedblleft bootstrap
principal\textquotedblright\ \cite{BK}.

\subsection{Energy momentum}

\label{s5.1}

The energy momentum tensor is in terms of fields is%
\[
T^{\mu\nu}(x)=\frac{1}{2}i\bar{\psi}\gamma^{\mu}\overleftrightarrow
{\partial^{\nu}}\psi-g^{\mu\nu}\mathcal{L}%
\]
with the trace%
\[
T_{\,\,\,\mu}^{\mu}(x)=m\bar{\psi}\psi\,.
\]
Because $T^{\mu\nu}$ is an $O(N)$ iso-scalar we have the weights
$w=(w_{1},\dots,w_{N/2})=(0,\dots,0)$ (see \cite{BFK5,BFK6}) which implies
that%
\[
n=n_{1}=\dots=n_{N/2-2}=n_{-}+n_{+},n_{-}=n_{+}\,.
\]
We write the energy momentum tensor in terms of an energy momentum potential
(see e.g. \cite{BFK7})%
\begin{align*}
T^{\mu\nu}(x)  &  =R^{\mu\nu}(i\partial_{x})T(x)\\
R^{\mu\nu}(P)  &  =-P^{\mu}P^{\nu}+g^{\mu\nu}P^{2}\\
T_{\,\,\,\mu}^{\mu}(x)  &  =(i\partial_{x})^{2}T(x)\,.
\end{align*}
For $\bar{\psi}\psi$ we propose the $n$-particle form factor as%
\begin{align}
\langle0|\bar{\psi}\psi(0)|\underline{\theta}\rangle_{\underline{\alpha}}  &
=F_{\underline{\alpha}}^{\bar{\psi}\psi}(\underline{\theta})=N_{n}^{\bar{\psi
}\psi}\prod_{i<j}F(\theta_{ij})K_{\underline{\alpha}}^{\bar{\psi}\psi
}(\underline{\theta})\nonumber\\
K_{\underline{\alpha}}^{\bar{\psi}\psi}(\underline{\theta})  &  =\int
_{\mathcal{C}_{\underline{\theta}}^{(1)}}dz_{1}\dots\int_{\mathcal{C}%
_{\underline{\theta}}^{(m)}}dz_{m}\,\tilde{h}(\underline{\theta},\underline
{z})p^{\bar{\psi}\psi}(\underline{\theta},\underline{z})\,\tilde{\Psi
}_{\underline{\alpha}}(\underline{\theta},\underline{z}) \label{EM}%
\end{align}
with $m=n=$ even and%
\begin{align}
\tilde{h}(\underline{\theta},\underline{z})  &  =\prod_{i=1}^{n}\prod
_{j=1}^{m}\tilde{\phi}(\theta_{i}-z_{j})\prod_{1\leq i<j\leq m}\tau
(z_{ij}),\nonumber\\
p^{\bar{\psi}\psi}(\underline{\theta},\underline{z})  &  =1\label{pT}\\
\tilde{\Psi}_{\underline{\alpha}}(\underline{\theta},\underline{z})  &
=L_{\underline{\mathring{\beta}}}(\underline{z})\,\Big(\Pi_{\underline{\beta}%
}^{\underline{\mathring{\beta}}}(\underline{z})\Omega\tilde{T}_{1}^{\beta_{m}%
}(\underline{\theta},z_{m})\dots\tilde{T}_{1}^{\beta_{1}}(\underline{\theta
},z_{1})\Big)_{\underline{\alpha}}\,.\nonumber
\end{align}

We do not calculate the integrals in (\ref{EM}) for general $N$, but the 2
particle form factor follows from lemma \ref{l2a} in appendix \ref{sf}%

\begin{equation}
F_{\alpha_{1}\alpha_{2}}^{\bar{\psi}\psi}(\underline{\theta})=\langle
\,0\,|\,\bar{\psi}\psi(0)\,|\,\theta_{1},\theta_{2}\,\rangle_{\alpha_{1}%
\alpha_{2}}^{in}=\mathbf{C}_{\alpha_{1}\alpha_{2}}\,\bar{v}(\theta
_{1})u(\theta_{2})\,F_{0}(\theta_{12})\label{EMN}%
\end{equation}
and%
\begin{align*}
F_{\alpha_{1}\alpha_{2}}^{T^{\mu\nu}}(\underline{\theta}) &  =\langle
\,0\,|\,T^{\mu\nu}(0)\,|\,\theta_{1},\theta_{2}\,\rangle_{\alpha_{1}\alpha
_{2}}^{in}=\mathbf{C}_{\alpha_{1}\alpha_{2}}\bar{v}(\theta_{1})\gamma^{\mu
}u(\theta_{2})\,\tfrac{1}{2}(p_{1}^{\nu}-p_{2}^{\nu})\,F_{0}(\theta_{12})\\
&  =\mathbf{C}_{\alpha_{1}\alpha_{2}}\bar{v}(\theta_{1})u(\theta_{2}%
)\,m\frac{\left(  p_{1}-p_{2}\right)  ^{\mu}(p_{1}^{\nu}-p_{2}^{\nu})}{\left(
p_{1}-p_{2}\right)  ^{2}}\,F_{0}(\theta_{12})\\
F_{\alpha_{1}\alpha_{2}}^{T}(\underline{\theta}) &  =\langle
\,0\,|\,T(0)\,|\,p_{1},p_{2}\,\rangle_{\alpha_{1}\alpha_{2}}^{in}%
=\mathbf{C}_{\alpha_{1}\alpha_{2}}\frac{\bar{v}(\theta_{1})u(\theta_{2}%
)}{4m\cosh^{2}\frac{1}{2}\theta_{12}}\,F_{0}(\theta_{12})
\end{align*}
with $F_{0}(\theta)$ given by (\ref{F0}) and (\ref{Ffull}).

\paragraph{1/N expansion:}

For $N\rightarrow\infty$ we obtain
\[
F_{\alpha_{1}\alpha_{2}}^{\bar{\psi}\psi}(\underline{\theta})=\mathbf{C}%
_{\alpha_{1}\alpha_{2}}\bar{v}(\theta_{2})u(\theta_{1})\frac{2\coth\frac{1}%
{2}\theta_{12}}{\theta_{12}-i\pi}+O(1/N)\,.
\]
This result agrees with the one obtained by computing Feynman graphs as was
done in \cite{KW}.

\subsection*{Conclusions:}

\addcontentsline{toc}{section}{Conclusions}

In this article we have enlarged our $O(N)$ Bethe Ansatz knowledge of the
$O(N)$ Gross-Neveu model, which exhibits a very rich bound state structure
and, consequently, creates a rich form factor hierarchy. We have computed the
form factors for the fundamental Fermi field, which transforms as a vector
representation of $O(N)$. Then we have also constructed the form factors for
the Noether current and the energy-momentum tensor. In addition for the two
particle case we have proved the recursion relation for the higher level
K-functions. Finally we have checked our results against the usual $1/N$
expansion and found full agreement. In a forthcoming paper we will investigate
the kink form factors, possibly proving a kink field equation. Moreover, we
will perform a detailed analysis of the $O(6)$ Gross-Neveu model, a starting
point in the nesting procedure.

\paragraph{Acknowledgment:}

The authors have profited from discussions with A. Fring, R. Schra\-der and B.
Schroer. H.B. thanks A. Belavin, A. Ferraz, V. Korepin, P. Sodano and P.
Wiegmann for valuable discussions. H.B. acknowledges financial support from
the Armenian grant 11-1c\_028 and the Armenian-Russian grant AR-17. He is also
grateful to the International Institute of Physics of UFRN (Natal) for
hospitality. A.F. acknowledges financial support from CNPq (Conselho Nacional
de Desenvolvimento Cientifico e Tecnologico). M.K. thanks J. Balog and P.
Weisz for discussions and hospitality at the Max-Planck Institut f\"{u}r
Physik (M\"{u}nchen), where parts of this work have been performed.

\appendix

\section*{Appendices}

\addcontentsline{toc}{part}{Appendices}

\renewcommand{\theequation}{\mbox{\Alph{section}.\arabic{equation}}} \setcounter{equation}{0}

\section{Proof of the main theorem \ref{TN}}

\label{sd}The identity%
\begin{equation}
\int_{\mathcal{C}_{a}}dz\Gamma(a-z)f(z)=2\pi i\,\operatorname*{Res}_{z=a}%
\sum_{l=-\infty}^{\infty}\Gamma(a-z-l)f(z+l) \label{Int}%
\end{equation}
where the $\mathcal{C}_{a}$ encircles the poles of $\Gamma(a-z)$
anti-clockwise may be used to write the K-function $K_{\underline{\alpha}%
}^{\mathcal{O}}(\underline{\theta})$ defined by the integral representation
(\ref{BA}) as a sum of \textquotedblleft Jackson-type Integrals" as
investigated in \cite{BFK5}. These expressions satisfy symmetry properties and
a matrix difference equation which are equivalent to the form factor equations
(i) and (ii). We have to prove, that due to the assumptions of theorem
\ref{TN} in addition the residue relations (iii)
\[
\operatorname*{Res}_{\theta_{12}=i\pi}F_{1\dots n}^{\mathcal{O}}(\theta
_{1},\dots,\theta_{n})=2i\,\mathbf{C}_{12}\,F_{3\dots n}^{\mathcal{O}}%
(\theta_{3},\dots,\theta_{n})\left(  \mathbf{1}-\sigma_{2}^{\mathcal{O}%
}\left(  \sigma S\right)  _{2n}\dots\left(  \sigma S\right)  _{23}\right)
\]
and (iv)%
\[
\operatorname*{Res}_{\theta_{12}=i\pi\nu}F_{12\dots n}^{\mathcal{O}}%
(\theta_{1},\theta_{2},\dots,\theta_{n})\,=F_{(12)\dots n}^{\mathcal{O}%
}(\theta_{(12)},\dots,\theta_{n})\,\sqrt{2}\Gamma_{12}^{(12)}%
\]
are satisfied.

\begin{proof}
We prove that the K-function $K_{1\dots n}^{\mathcal{O}}(\underline{\theta})$
defined by (\ref{2.10}) and (\ref{BA}) satisfies the form factor equations (i)
- (iii) which read in terms of $K_{1\dots n}^{\mathcal{O}}(\underline{\theta
})$ as%
\begin{gather}
K_{\dots ij\dots}^{\mathcal{O}}(\dots,\theta_{i},\theta_{j},\dots)=K_{\dots
ji\dots}^{\mathcal{O}}(\dots,\theta_{j},\theta_{i},\dots)\,\tilde{S}%
_{ij}(\theta_{ij})\label{2.12}\\
K_{1\ldots n}^{\mathcal{O}}(\theta_{1}+2\pi i,\theta_{2},\dots,\theta
_{n})\sigma_{1}^{\mathcal{O}}\mathbf{C}^{\bar{1}1}=K_{2\ldots n1}%
^{\mathcal{O}}(\theta_{2},\dots,\theta_{n},\theta_{1})\mathbf{C}^{1\bar{1}%
}\label{2.14}\\
\operatorname*{Res}_{\theta_{12}=i\pi}K_{1\dots n}(\underline{\theta}%
)=\frac{2i\,}{F(i\pi)}\mathbf{C}_{12}\prod_{i=3}^{n}\tilde{\phi}(\theta
_{i1}+i\pi\nu)\tilde{\phi}(\theta_{i2})K_{3\dots n}(\theta_{3},\dots
,\theta_{n})\left(  \mathbf{1}-\sigma_{2}^{\mathcal{O}}S_{2n}\dots
S_{23}\right)  \label{2.15}%
\end{gather}
where (\ref{FF}) has been used.

\begin{itemize}
\item (i) follows as in \cite{BFK5},

\item (ii) follows as in \cite{BFK5}, however, here (\ref{p}) is responsible
for the statistics factor $\sigma_{1}^{\mathcal{O}}$ in (\ref{2.14}).

\item (iii) the residue of
\begin{equation}
K_{1\dots n}^{\mathcal{O}}(\underline{\theta})=N_{n}^{\mathcal{O}}%
\int_{\mathcal{C}_{\underline{\theta}}^{(1)}}dz_{1}\cdots\int_{\mathcal{C}%
_{\underline{\theta}}^{(m)}}dz_{m}\,\tilde{h}(\underline{\theta},\underline
{z})\,p^{\mathcal{O}}(\underline{\theta},\underline{z})\,\,\tilde{\Psi
}_{1\dots n}(\underline{\theta},\underline{z}) \label{Kn}%
\end{equation}
consists of two terms
\[
\operatorname*{Res}_{\theta_{12}=i\pi}K_{1\dots n}(\underline{\theta
})=\bigg(\operatorname*{Res}_{\theta_{12}=i\pi}^{(1)}+\operatorname*{Res}%
_{\theta_{12}=i\pi}^{(2)}\bigg)K_{1\dots n}(\underline{\theta})\,.
\]
This is because for each $z_{j}$ integration with $j$ even the contours will
be \textquotedblleft pinched\textquotedblright\ at two points (see Fig.
\ref{f5.1}):
\end{itemize}

\begin{itemize}
\item[(1)] $z_{j}=\theta_{2}\approx\theta_{1}-i\pi$

\item[(2)] $z_{j}=\theta_{1}-2\pi i\approx\theta_{2}-i\pi$
\end{itemize}

We prove in appendix \ref{se} the residue formulas for general level $k$ of
the off-shell Bethe Ansatz. In particular for $k=0$ the general result implies
that the pinching (1) gives%
\begin{equation}
\operatorname*{Res}_{\theta_{12}=i\pi}^{(1)}K_{1\dots n}(\underline{\theta
})=\frac{2i\,}{F(i\pi)}\mathbf{C}_{12}\prod_{i=3}^{n}\tilde{\phi}(\theta
_{i1}+i\pi\nu)\tilde{\phi}(\theta_{i2})K_{3\dots n}(\theta_{3},\dots
,\theta_{n}) \label{RK}%
\end{equation}
for a suitable choice of the normalization constants in (\ref{Kn}). Therefore
we have proved%
\[
\operatorname*{Res}_{\theta_{12}=i\pi}^{(1)}F_{1\dots n}(\theta_{1}%
,\dots,\theta_{n})=2i\,\mathbf{C}_{12}\,F_{3\dots n}(\theta_{3},\dots
,\theta_{n})\,.
\]
We use (ii) and (i) to write
\begin{align*}
F_{1\ldots n}(\underline{\theta})\sigma_{1}^{\mathcal{O}}  &  =\mathbf{C}%
_{1\bar{1}}F_{2\ldots n1}(\theta_{2},\dots,\theta_{n},\theta_{1}-2\pi
i)\mathbf{C}^{1\bar{1}}\\
&  =\mathbf{C}_{1\bar{1}}F_{21\ldots n}(\theta_{2},\theta_{1}-2\pi
i,\dots,\theta_{n})\mathbf{C}^{1\bar{1}}\left(  \sigma S\right)  _{\bar{1}%
n}\dots\left(  \sigma S\right)  _{\bar{1}3}\,.
\end{align*}
Then the result for $\operatorname*{Res}\limits_{\theta_{1}=\theta_{2}+i\pi
}^{(1)}$ implies for the contribution of the pinching at $z_{j}=\theta
_{1}-2\pi i\approx\theta_{2}-i\pi$
\begin{align*}
\operatorname*{Res}_{\theta_{1}=\theta_{2}+i\pi}^{(2)}F_{1\ldots n}%
(\underline{\theta})\sigma_{1}^{\mathcal{O}}  &  =-\operatorname*{Res}%
_{\theta_{2}=\left(  \theta_{1}-2\pi i\right)  +i\pi}^{(1)}\mathbf{C}%
_{1\bar{1}}F_{21\ldots n}(\theta_{2},\theta_{1}-2\pi i,\dots,\theta
_{n})\mathbf{C}^{1\bar{1}}\left(  \sigma S\right)  _{\bar{1}n}\dots\left(
\sigma S\right)  _{\bar{1}3}\\
&  =-\mathbf{C}_{1\bar{1}}2i\,\mathbf{C}_{21}\,F_{3\dots n}^{\mathcal{O}%
}(\theta_{3},\dots,\theta_{n})\mathbf{C}^{1\bar{1}}\left(  \sigma S\right)
_{\bar{1}n}\dots\left(  \sigma S\right)  _{\bar{1}3}\\
&  =-2i\mathbf{C}_{12}\,\,F_{3\dots n}^{\mathcal{O}}(\theta_{3},\dots
,\theta_{n})\sigma_{2}^{\mathcal{O}}\left(  \sigma S\right)  _{2n}\dots\left(
\sigma S\right)  _{23}\sigma_{1}^{\mathcal{O}}%
\end{align*}
using $\sigma_{1}^{\mathcal{O}}\sigma_{\bar{1}}^{\mathcal{O}}=1$.

(iv) Because there are bound states we also have to discuss the form factor
equation (iv) (\ref{1.16})%
\[
\operatorname*{Res}_{\theta_{12}=i\pi\nu}F_{12\dots n}^{\mathcal{O}}%
(\theta_{1},\theta_{2},\underline{\hat{\theta}})\,=F_{(12)\dots n}%
^{\mathcal{O}}(\theta_{(12)},\underline{\hat{\theta}})\,\sqrt{2}\Gamma
_{12}^{(12)}.
\]
The bound state form factor $F_{(12)\dots n}^{\mathcal{O}}(\theta
_{(12)},\underline{\hat{\theta}})\,$is then obtained from the residue%
\[
\operatorname*{Res}_{\theta_{12}=i\pi\nu}K_{12\dots n}^{\mathcal{O}%
}(\underline{\theta})=\operatorname*{Res}_{\theta_{12}=i\pi\nu}N_{n}%
^{\mathcal{O}}\int_{\mathcal{C}_{\underline{\theta}}^{(1)}}dz_{1}\cdots
\int_{\mathcal{C}_{\underline{\theta}}^{(m)}}dz_{m}\,\tilde{h}(\underline
{\theta},\underline{z})\,p^{\mathcal{O}}(\underline{\theta},\underline
{z})\,\,\tilde{\Psi}_{1\dots n}(\underline{\theta},\underline{z})\,.
\]
Similar as in the proof of (iii) the residue is obtained from pinching at:

$z_{j}=\theta_{1}-i\pi\nu\approx\theta_{2}$ for $\mathcal{C}^{(o)}$ and
$z_{j}=\theta_{2}\approx\theta_{1}-i\pi\nu$ for $\mathcal{C}^{(e)}$.

Here we will not perform the lengthy calculations and write the complicated
result, but in appendix \ref{siv} we will calculate the bound state form
factors for the examples of section \ref{s5}.
\end{proof}

\section{Two-particle current form factor}

\label{sh}

Derivation of (\ref{FJ}) and (\ref{FJmu}):

\begin{proof}
The two-particle K-function of the current is
\[
K_{\underline{\alpha}}^{J}(\underline{\theta})=N_{2}^{J}\int_{\mathcal{C}%
_{\underline{\theta}}^{(o)}}dz\,\tilde{h}(\underline{\theta},\underline
{z})p^{J}(\underline{\theta},z)\,\tilde{\Psi}_{\underline{\alpha}}%
(\underline{\theta},z)
\]
with the p-function (\ref{pJ}) for $n=2$ and $m=1$%
\begin{equation}
p^{J}(\underline{\theta},\underline{\underline{z}})=\frac{e^{\frac{1}%
{2}\left(  \theta_{1}+\theta_{2}\right)  }}{e^{\theta_{1}}+e^{\theta_{2}}%
}+\frac{e^{-\frac{1}{2}\left(  \theta_{1}+\theta_{2}\right)  }}{e^{-\theta
_{1}}+e^{-\theta_{2}}}=\frac{1}{\cosh\frac{1}{2}\theta_{12}} \label{pJ2}%
\end{equation}
and the Bethe state%
\[
\tilde{\Psi}_{\underline{\alpha}}(\underline{\theta},z)=\delta_{\alpha_{1}%
}^{2}\delta_{\alpha_{2}}^{1}\tilde{c}(\theta_{1}-z)+\delta_{\alpha_{1}}%
^{1}\delta_{\alpha_{2}}^{2}\tilde{b}(\theta_{1}-z)\tilde{c}(\theta_{2}-z)\,.
\]
Doing the integral we obtain%
\begin{equation}
K_{21}^{J}(\theta_{12})=N_{2}^{J}\int_{\mathcal{C}_{\underline{\theta}}^{(o)}%
}dz\,\tilde{h}(\underline{\theta},z)p^{J}(\underline{\theta},z)\,\tilde{\Psi
}_{21}(\underline{\theta},z)=-N_{2}^{J}8\pi^{3}4^{\nu}\Gamma\left(
1-\nu\right)  c\frac{F_{-}\left(  \theta\right)  }{F\left(  \theta\right)  }
\label{Kj2}%
\end{equation}
and $K_{12}^{J}(\theta)=-K_{21}^{J}(\theta).$

We use again the variables $u=\theta/(i\pi\nu)$ and $v=z/(i\pi\nu)$, consider
the component $K_{21}^{J}(\theta)$ and calculate the integral%
\begin{align*}
I  &  =\frac{1}{2\pi i}\int_{\mathcal{C}_{_{\underline{u}}}^{(o)}%
}dv\,I(\underline{u},v)\\
I(\underline{u},v)  &  =\tilde{h}(\underline{u},v)\tilde{\Psi}(\underline
{u},v),~\tilde{h}(\underline{u},v)=\tilde{\phi}(u_{1}-v)\tilde{\phi}%
(u_{2}-v),~\tilde{\Psi}(\underline{u},v)=\tilde{c}(u_{1}-v)\,.
\end{align*}
Writing the integrals in terms of sums over residues we obtain (see Fig.
\ref{f5.1})%
\begin{align}
I  &  =I_{1}+I_{2}=\sum_{l=0}^{\infty}s_{1}(u_{1},u_{2},l)+\sum_{l=0}^{\infty
}s_{2}(u_{1},u_{2},l)\label{Ix}\\
s_{i}(u_{1},u_{2},l)  &  =\operatorname*{Res}_{v=v_{o}(u_{i},l)}I(u_{1}%
,u_{2},v),~v_{o}(u,l)=u-1+2l/\nu\,.\nonumber
\end{align}
Using the Gauss formula%
\begin{equation}
_{2}F_{1}(a,b;c;1)=\sum_{n=0}^{\infty}\frac{\Gamma\left(  a+n\right)  }%
{\Gamma\left(  a\right)  }\frac{\Gamma\left(  b+n\right)  }{\Gamma\left(
b\right)  }\frac{\Gamma\left(  c\right)  }{\Gamma\left(  c+n\right)  }%
\frac{\Gamma\left(  1\right)  }{\Gamma\left(  1+n\right)  }=\frac
{\Gamma\left(  c\right)  \Gamma\left(  c-a-b\right)  }{\Gamma\left(
c-a\right)  \Gamma\left(  c-b\right)  } \label{Gauss}%
\end{equation}
we get
\begin{align*}
I  &  =I_{1}+I_{2}\\
&  =2^{\nu}\pi\sqrt{\pi}\frac{\Gamma\left(  -\frac{1}{2}\nu\right)
\Gamma\left(  \frac{1}{2}\nu+\frac{1}{2}\right)  \cos\frac{1}{2}\pi\nu}%
{\sin\frac{1}{2}\pi\nu\left(  u_{12}+1\right)  \sin\frac{1}{2}\pi\nu\left(
u_{12}-1\right)  }\frac{1}{\Gamma\left(  \frac{1}{2}\nu+\frac{1}{2}\nu
u_{12}\right)  \Gamma\left(  1+\frac{1}{2}\nu-\frac{1}{2}\nu u_{12}\right)  }%
\end{align*}
which agrees with (\ref{Kj2}) taking (\ref{pJ2}) into account. Therefore using
(\ref{Fmin}) and (\ref{Ffull}) we finally obtain with the normalization
constant%
\begin{equation}
N_{2}^{J}=\frac{1}{8}4^{-\nu}\frac{im}{c\pi^{3}\Gamma\left(  1-\nu\right)  }
\label{NJ2}%
\end{equation}
for the pseudo-potential $J^{\alpha\beta}(x)$ the two-particle form factors
(\ref{FJ}) and (\ref{FJmu}) for the current. The normalization is chosen such
that the form factor agrees for $F_{-}\left(  \theta\right)  \rightarrow
F_{-}\left(  i\pi\right)  =1$ with the free field expression.
\end{proof}

\section{Higher level K-functions}

\subsection{Proof of lemma \ref{L1}}

\label{se}

\begin{remark}
If in (\ref{iiik}) $\operatorname*{Res}$ is replaced by $\overset
{(1)}{\operatorname*{Res}}$ Lemma \ref{L1} also holds for $k=0$ as explained
in appendix \ref{sd}.
\end{remark}

For the discussion of the general $k$-level Bethe Ansatz it is convenient to
use the variables $u,v$ defined by $\theta=i\pi\nu_{k}u,~z=i\pi\nu_{k}v$ and
$\nu_{k}=2/(N-2k-2)$ (for the S-matrix see (\ref{Su})). In the proof we will
replace $p^{(k)}(\underline{u},\underline{v})$ by $1$ which will not change
the results, if the $p^{(k)}$ satisfy the conditions (\ref{pk}).

\begin{proof}
As above in the proof of theorem \ref{TN} the relations (i)$^{(k)}$ and
(ii)$^{(k)}$ follow from the results of \cite{BFK5}. The proof of
(iii)$^{(k)}$ is the same as the corresponding one in \cite{BFK7}, only the
functions $\tilde{\psi}(u)$ and $\tilde{\chi}(u)$ have to replaced by
$\tilde{\phi}(u)$ and $\tau_{ij}(v)$ by $\tau(v)$. As in \cite{BFK7} one
finally obtains%
\begin{align*}
&  \operatorname*{Res}_{u_{12}=1/\nu_{k}}K_{\underline{\alpha}}^{(k)}%
(\underline{u})\\
&~~  =const.\left(  \operatorname*{Res}_{v=1/\nu_{k+1}}\tilde{d}_{k+1}%
(v)\right)  ^{-1}\operatorname*{Res}_{u_{12}=1/\nu_{k}}\oint_{u_{1}-1}%
dv_{1}\tilde{c}(u_{1}-v_{1})\left(  -\oint_{u_{2}}\right)  dv_{2}\\
&~~~~~  \times\tilde{d}_{k}(u_{12})\left(  \prod_{i=1}^{2}\prod_{j=1}^{2}%
\tilde{\phi}(u_{i}-v_{j})\right)  \tau(v_{12})\prod_{i=3}^{n_{k}}\tilde{\phi
}(u_{i1}+1)\tilde{\phi}(u_{i2})\frac{1}{\tilde{N}_{m_{k}-2}^{(k)}}%
\mathbf{C}_{\alpha_{1}\alpha_{2}}K_{\underline{\check{\alpha}}}^{(k)}%
(\underline{\check{u}})\,
\end{align*}
with $\underline{\check{\alpha}}=\alpha_{3}\ldots\alpha_{n_{k}}$. It has been
used that for $u_{12}=1/\nu_{k},~v_{12}=1/\nu_{k+1},~u_{2}=v_{2},~u_{1}%
=v_{2}+1/\nu_{k}=v_{1}+1$%
\[
\left(  \frac{a_{k+1}(v_{1j})a_{k+1}(v_{2j})}{a_{k}(u_{1}-v_{j})a_{k}%
(u_{2}-v_{j})}\tilde{\phi}(v_{j1}+1)\tilde{\phi}(v_{j2})\right)  \left(
\tilde{\phi}(u_{1}-v_{j})\tilde{\phi}(u_{2}-v_{j})\tau(v_{1j})\tau
(v_{2j})\right)  =1\,.
\]
This can be shown by means of (\ref{shiftphi}) and the formulas
\begin{align*}
a_{k}(u_{1})a_{k}(u_{2}) &  =\tilde{b}(-u_{2})/\tilde{b}(u_{1})\\
\tilde{b}(u)\tilde{\phi}(u) &  =-\tilde{\phi}(1-u)\,.
\end{align*}
The final result is that equation (\ref{iiik}) holds for a suitable choice of
the normalization constants in (\ref{Kk}).
\end{proof}

\subsection{Two-particle higher level K-functions}

\label{sf}We need higher level K-functions for the examples of section
\ref{s5}, in particular, in the iso-scalar two-particle channel (with weights
$w=(0,\dots,0)$) the K-function $K_{\alpha_{1}\alpha_{2}}^{(k)}(\theta
_{1},\theta_{2})$ (level $k=0,1,2,\dots$) belonging to $O(N-2k)$ . It is of
the form%
\begin{equation}
K_{\alpha_{1}\alpha_{2}}^{(k)}(u_{1},u_{2})=\mathbf{C}_{\alpha_{1}\alpha_{2}%
}^{(N-2k)}K(u_{12},k) \label{k0}%
\end{equation}
where $\mathbf{C}_{\alpha_{1}\alpha_{2}}^{(N-2k)}$ is the $O(N-2k)$ charge
conjugation matrix\footnote{In the real basis this would be $\delta
_{\alpha_{1}\alpha_{2}}$.}. From the weight vector%
\[
w=(w_{1},\dots,w_{N/2})=(0,\dots,0)=\left(  n-n_{1},\dots,n_{N/2-2}%
-n_{-}-n_{+},n_{-}-n_{+}\right)
\]
follows that for all levels $n_{k}=2$.

\begin{lemma}
\label{l2a}The vector valued functions $K_{\alpha_{1}\alpha_{2}}^{(k)}%
(u_{1},u_{2})$ with
\begin{equation}
K(u,k)=\frac{\Gamma\left(  1-\frac{1}{2}\nu-\frac{1}{2}\nu u\right)
\Gamma\left(  -\frac{1}{2}\nu+\frac{1}{2}\nu u\right)  }{\Gamma\left(
\frac{3}{2}-\frac{1}{2}k\nu-\frac{1}{2}\nu u\right)  \Gamma\left(  \frac{1}%
{2}-\frac{1}{2}k\nu+\frac{1}{2}\nu u\right)  } \label{K}%
\end{equation}
satisfy for $k=0,1,2,\dots<N/2-2$ the recursion relation%
\begin{align}
K_{\underline{\alpha}}^{(k)}\left(  \underline{u}\right)   &  =N^{(k)}%
\int_{\mathcal{C}_{\underline{u}}^{(o)}}dv_{1}\int_{\mathcal{C}_{\underline
{u}}^{(e)}}dv_{2}\,\tilde{h}(\underline{u},\underline{v})\,L_{\underline
{\mathring{\beta}}}^{(k)}(\underline{v},k)\tilde{\Phi}^{(k)}\,_{\underline
{\alpha}}^{\underline{\mathring{\beta}}}\,(\underline{u},\underline
{v})\label{rec}\\
L_{\underline{\mathring{\beta}}}^{(k)}(\underline{v},k)  &  =K_{\underline
{\mathring{\beta}}}^{(k+1)}(\underline{v})=\mathbf{C}_{\underline
{\mathring{\beta}}}^{(N-2k-2)}K(v_{12},k+1)\nonumber
\end{align}
with%
\begin{align*}
\tilde{h}(\underline{u},\underline{v})  &  =\prod_{i=1}^{2}\left(  \tilde
{\phi}(u_{i}-v_{1})\tilde{\phi}(u_{i}-v_{2})\right)  \frac{1}{\tilde{\phi
}(v_{12})\tilde{\phi}(-v_{12})}\\
\tilde{\Phi}^{(k)}\,_{\underline{\alpha}}^{\underline{\mathring{\beta}}%
}\,(\underline{u},\underline{v})  &  =\Big(\Pi_{\underline{\beta}}%
^{\underline{\mathring{\beta}}}(\underline{v})\Omega\tilde{T}_{1}^{\beta_{2}%
}(\underline{u},v_{2})\tilde{T}_{1}^{\beta_{1}}(\underline{u},v_{1}%
)\Big)_{\underline{\alpha}}^{(k)}%
\end{align*}
and the normalization%
\begin{equation}
N^{(k)}=\frac{-2^{-\nu}}{8\pi^{2}}\frac{\Gamma\left(  \frac{1}{2}-\frac{1}%
{2}k\nu\right)  \Gamma\left(  1-\frac{1}{2}k\nu-\frac{1}{2}\nu\right)
}{\Gamma\left(  1-\frac{1}{2}k\nu\right)  \Gamma\left(  \frac{1}{2}-\frac
{1}{2}k\nu+\frac{1}{2}\nu\right)  \left(  \Gamma\left(  -\frac{1}{2}%
\nu\right)  \right)  ^{2}}\,. \label{Nk}%
\end{equation}

\end{lemma}

\begin{proof}
The function (\ref{K}) satisfies%
\begin{equation}%
\begin{array}
[c]{ccrcl}%
\text{(i)}: &  & K(u,k) & = & K(-u,k)\tilde{S}_{0}^{(k)}(u)\\
\text{(ii)}: &  & K(1/\nu-u,k) & = & K(1/\nu+u,k)
\end{array}
\label{i+ii}%
\end{equation}
with the scalar eigenvalue of $\tilde{S}^{(k)}(u)=\tilde{S}^{O(N-2k)}(u)$
\[
\tilde{S}_{0}^{(k)}(u)=S_{0}^{(k)}(u)/S_{+}^{(k)}(u)=\frac{u+1/\nu_{k}%
}{u-1/\nu_{k}}\frac{u+1}{u-1}=\frac{u+\left(  1/\nu-k\right)  }{u-\left(
1/\nu-k\right)  }\frac{u+1}{u-1}\,.
\]
The minimal solution (with no poles in the physical strip $0\leq
\operatorname{Re}u\leq1/\nu$) is
\begin{equation}
K_{m}(u,k)=\frac{1}{\Gamma\left(  \frac{3}{2}-\frac{1}{2}k\nu-\frac{1}{2}\nu
u\right)  \Gamma\left(  \frac{1}{2}\left(  1-\nu k\right)  +\frac{1}{2}\nu
u\right)  \Gamma\left(  1+\frac{1}{2}\nu-\frac{1}{2}\nu u\right)
\Gamma\left(  \frac{1}{2}\nu+\frac{1}{2}\nu u\right)  } \label{Km}%
\end{equation}
whereas $K(u,k)$ has the bound state pole at $u=1~(\theta=i\pi\nu_{k})$.

The Bethe state in (\ref{rec}) is (see \cite{BFK5})%
\begin{align*}
\tilde{\Phi}^{(k)}\,_{\underline{\alpha}}^{\underline{\mathring{\beta}}%
}\,(\underline{u},\underline{v})  &  =\left(  \Pi^{(k)}\right)  _{\underline
{\beta}}^{\underline{\mathring{\beta}}}(\underline{v})\Big(\Omega\tilde{T}%
_{1}^{\beta_{2}}(\underline{u},v_{2})\tilde{T}_{1}^{\beta_{1}}(\underline
{u},v_{1})\Big)_{\underline{\alpha}}^{(k)}\\
\left(  \Pi^{(k)}\right)  _{\beta_{1}\beta_{2}}^{\mathring{\beta}_{1}%
\mathring{\beta}_{2}}(\underline{v})  &  =\delta_{\beta_{1}}^{\mathring{\beta
}_{1}}\delta_{\beta_{2}}^{\mathring{\beta}_{2}}+f_{k}(v_{12})\mathbf{\mathring
{C}}^{\mathring{\beta}_{1}\mathring{\beta}_{2}}\delta_{\beta_{1}}%
^{\overline{k+1}}\delta_{\beta2}^{k+1}\,,~~f_{k}(v)\,=\frac{1}{v+1/\nu_{k}-1}%
\end{align*}
which may be depicted (for $k=0$) as%
\begin{equation}
\Phi_{\underline{\alpha}}^{\underline{\mathring{\beta}}}(\underline
{u},\underline{w})=~~~%
\begin{array}
[c]{c}%
\unitlength4mm\begin{picture}(10,8)(0,1) \thicklines\put(4,1.8){$u_1$} \put(8.3,1.8){$u_2$} \put(3.3,4.2){$v_2$} \put(2,2){$v_1$} \put(5.0,6.2){1} \put(8.0,6.2){1} \put(9.5,2.8){1} \put(9.5,4.8){1} \put(5,2){\line(0,1){4}} \put(8,2){\line(0,1){4}} \put(9,6){\oval(18,6)[lb]} \put(9,6){\oval(13,2)[lb]} \put(-.2,6){$\framebox(3,1.5){$\Pi$}$} \put(1.,8.2){$\underline{\mathring{\beta}}$} \put(6,1){$\underline{\alpha}$} \put(0,7.5){\line(0,1){1}} \put(2.5,7.5){\line(0,1){1}} \end{picture}
\end{array}
~~,~~~\Pi_{\alpha\beta}^{\mathring{\alpha}\mathring{\beta}}(u_{1},u_{2})=%
\begin{array}
[c]{c}%
\unitlength3.6mm\begin{picture}(2.5,6)\thicklines \put(0,1){\line(0,1){4}} \put(-.27,2.7){$\bullet$} \put(2,1){\line(0,1){4}} \put(1.73,2.7){$\bullet$} \put(-.2,0){$\alpha$} \put(1.8,0){$\beta$} \put(-.2,5.4){$\mathring{\alpha}$} \put(1.8,5.4){$\mathring{\beta}$} \end{picture}
\end{array}
+f(u_{12})~~%
\begin{array}
[c]{c}%
\unitlength3.6mm\begin{picture}(3.5,6)\thicklines \put(0,1){\line(0,1){2}} \put(2,1){\line(0,1){2}} \put(1,4){\line(1,1){1}} \put(1,4){\line(-1,1){1}} \put(-.2,0){$\alpha$} \put(1.8,0){$\beta$} \put(-.2,3.3){$\bar1$} \put(1.8,3.3){$1$} \put(-.2,5.4){$\mathring{\alpha}$} \put(1.8,5.4){$\mathring{\beta}$} \end{picture}
\end{array}
\,. \label{psi}%
\end{equation}
Because of (\ref{k0}) it is sufficient to consider only one component of
(\ref{rec}) and for convenience we take $\tilde{\Phi}^{(k)}\,_{\underline
{\alpha}}^{\underline{\mathring{\beta}}}$ with $\underline{\alpha}%
=\overline{k+1},k+1$ and define%
\begin{align*}
& \tilde{\Phi}_{k}(\underline{u},\underline{v})\\
& =\mathbf{C}_{\mathring{\beta}_{1}\mathring{\beta}_{2}}^{(N-2k-2)}\tilde
{\Phi}^{(k)}\,_{\overline{k+1},k+1}^{\mathring{\beta}_{1}\mathring{\beta}_{2}%
}\,(\underline{u},\underline{v})\\
& =\mathbf{C}_{\mathring{\beta}_{1}\mathring{\beta}_{2}}^{(N-2k-2)}%
\mathbf{C}_{(N-2k-2)}^{\mathring{\beta}_{1}\mathring{\beta}_{2}}\left(
\tilde{c}(u_{1}-v_{2})\tilde{d}_{k}(u_{1}-v_{1})+f_{k}(v_{1}-v_{2})\left(
\tilde{c}(u_{1}-v_{1})+\tilde{d}_{k}(u_{1}-v_{1})\right)  \right)  \\
& =(N-2k-2)\frac{-\left(  v_{1}-v_{2}-1\right)  }{\left(  u_{1}-v_{1}%
-1\right)  \left(  u_{1}-v_{2}-1\right)  \left(  v_{1}-v_{2}+1/\nu
_{k}-1\right)  }%
\end{align*}
where (for $k=0$) $\tilde{\Phi}\,_{\overline{1},1}^{\mathring{\beta}%
_{1}\mathring{\beta}_{2}}(\underline{u},\underline{v})$ may be depicted as%
\[%
\begin{array}
[c]{c}%
\unitlength3.5mm\begin{picture}(8,7)(16,4.5) \thicklines \put(18.8,4.3){$u_1$} \put(22.3,4.3){$u_2$} \put(20,3){$\bar 1$} \put(22,3){$1$} \put(16,10.6){$\mathring{\beta}_1$} \put(17.5,10.6){$\mathring{\beta}_2$} \put(17,5.4){$v_{2}$} \put(18.6,7.3){$v_1$} \put(24.2,5.6){1} \put(24.2,6.6){1} \put(19.7,8.1){1} \put(21.4,8.1){1} \put(21,4.5){\oval(2,3)[t]} \put(18.5,10){\oval(1,6)[lb]} \put(18.5,10){\oval(4,8)[lb]} \put(24,8){\oval(4,4)[lb]} \put(18.5,6.5){\oval(2,1)[r]} \put(24,8){\oval(8,2)[lb]} \end{picture}
\end{array}
~~~~+f(v_{12})~\Bigg(~~%
\begin{array}
[c]{c}%
\unitlength3.5mm\begin{picture}(8,7)(16,4.5) \thicklines \put(18.8,4.3){$u_1$} \put(22.3,4.3){$u_2$} \put(20,3){$\bar 1$} \put(22,3){$1$} \put(16,10.6){$\mathring{\beta}_1$} \put(17.5,10.6){$\mathring{\beta}_2$} \put(17.,10){\oval(1.5,2)[b]} \put(16.2,7.6){$\bar1$} \put(17.8,7.6){$1$} \put(17,5.1){$v_{2}$} \put(18.6,7.3){$v_1$} \put(24.2,5.6){1} \put(24.2,6.6){1} \put(19.7,8.1){1} \put(21.4,8.1){1} \put(21,4.5){\oval(2,3)[rt]} \put(21,6.5){\oval(2.,1)[lb]} \put(18.5,7.5){\oval(1,1)[lb]} \put(18.5,7.5){\oval(4,3)[lb]} \put(24,8){\oval(4,4)[lb]} \put(18.5,6.5){\oval(3,1)[rt]} \put(18.5,4.5){\oval(2.7,3)[rt]} \put(24,8){\oval(8,2)[lb]} \end{picture}
\end{array}
~~~+~%
\begin{array}
[c]{c}%
\unitlength3.5mm\begin{picture}(8,7)(16,4.5) \thicklines \put(18.8,4.3){$u_1$} \put(22.3,4.3){$u_2$} \put(20,3){$\bar 1$} \put(22,3){$1$} \put(16,10.6){$\mathring{\beta}_1$} \put(17.5,10.6){$\mathring{\beta}_2$} \put(17.,10){\oval(1.5,2)[b]} \put(16.2,7.6){$\bar1$} \put(17.8,7.6){$1$} \put(17,5.1){$v_{2}$} \put(18.6,7.3){$v_1$} \put(24.2,5.6){1} \put(24.2,6.6){1} \put(19.7,8.1){1} \put(21.4,8.1){1} \put(21,4.5){\oval(2,3)[t]} \put(18.5,7.5){\oval(1,1)[lb]} \put(18.5,7.5){\oval(4,3)[lb]} \put(24,8){\oval(4,4)[lb]} \put(18.5,6.5){\oval(2,1)[r]} \put(24,8){\oval(8,2)[lb]} \end{picture}
\end{array}
~~~\Bigg)\,.
\]
Then because $L_{\underline{\mathring{\beta}}}^{(k)}(\underline{v})\tilde
{\Phi}^{(k)}\,_{\overline{k+1},k+1}^{\underline{\mathring{\beta}}%
}\,(\underline{u},\underline{v})=K(v_{12},k+1)\tilde{\Phi}_{k}(\underline
{u},\underline{v})$
\begin{align}
K_{\overline{k+1},k+1}^{(k)}\left(  \underline{u}\right)   &  =N^{(k)}%
\int_{\mathcal{C}_{\underline{u}}^{(o)}}dv_{1}\int_{\mathcal{C}_{\underline
{u}}^{(e)}}dv_{2}\,\tilde{h}(\underline{u},\underline{v})\,K(v_{12},k+1)\tilde{\Phi}_{k}(\underline{u},\underline{v})\label{Kk1}\\
&  =N^{(k)}(2\pi i)^{2}(N-2k-2)J(u_{12})\nonumber
\end{align}
where%
\begin{align}
J(u_{12})  &  =\frac{1}{(2\pi i)^{2}}\int_{\mathcal{C}_{\underline{u}}^{(o)}%
}dv_{1}\int_{\mathcal{C}_{\underline{u}}^{(e)}}dv_{2}\,J(\underline
{u},\underline{v})\label{Ju}\\
J(\underline{u},\underline{v})  &  =\tilde{\phi}(u_{1}-v_{1})\tilde{c}%
(u_{1}-v_{1})\tilde{\phi}(u_{1}-v_{2})\tilde{c}(u_{1}-v_{2})\tilde{\phi}%
(u_{2}-v_{1})\tilde{\phi}(u_{2}-v_{2})\varphi(v_{12})\nonumber\\
\varphi(v)  &  =\frac{\left(  1-v\right)  K(v,k+1)}{\tilde{\phi}(v)\tilde
{\phi}(-v)\left(  v+1/\nu-k-1\right)  }.\nonumber
\end{align}
Because the function $J(u)$ satisfies (\ref{i+ii}) it is proportional to
$K(u,k)$ (as was shown in general in \cite{KW}) if there are no zeroes and
exactly one pole\footnote{This is suggested by numerical calculations using
mathematica.} at $u=1$ in $0\leq\operatorname{Re}u\leq1/\nu$. Finally we
obtain%
\begin{equation}
J(u)=\frac{2^{\nu-1}\nu\left(  \Gamma\left(  -\frac{1}{2}\nu\right)  \right)
^{2}\Gamma\left(  \frac{1}{2}-\frac{1}{2}k\nu+\frac{1}{2}\nu\right)
\Gamma\left(  1-\frac{1}{2}k\nu\right)  }{\Gamma\left(  -\frac{1}{2}%
\nu+1-\frac{1}{2}k\nu\right)  \Gamma\left(  \frac{3}{2}-\frac{1}{2}%
k\nu\right)  }\,K(u,k) \label{JK}%
\end{equation}
where the constant is calculated by taking the residue at $u=1$ on both sides
of (\ref{JK}). Finally we turn to (\ref{rec}). By (\ref{Kk1}) and (\ref{k0})
we have
\begin{equation}
K(u,k)=K_{\overline{k+1},k+1}^{(k)}\left(  \underline{u}\right)  =N^{(k)}(2\pi
i)^{2}(N-2k-2)J(u_{12})\, \label{rec1}%
\end{equation}
which provides the normalization (\ref{Nk}).
\end{proof}

In particular for $k=0$ and $k=1$%

\begin{align*}
K(u)  &  =K(u,0)=\frac{-2}{\pi}\frac{\cos\frac{1}{2}\pi\nu u}{\nu u-1}%
\Gamma\left(  -\tfrac{1}{2}\nu+\tfrac{1}{2}\nu u\right)  \Gamma\left(
1-\tfrac{1}{2}\nu-\tfrac{1}{2}\nu u\right) \\
L(u)  &  =K(u,1)=\frac{\Gamma\left(  1-\frac{1}{2}\nu-\frac{1}{2}\nu u\right)
\Gamma\left(  -\frac{1}{2}\nu+\frac{1}{2}\nu u\right)  }{\Gamma\left(
1+\frac{1}{2}\left(  1-\nu\right)  -\frac{1}{2}\nu u\right)  \Gamma\left(
\frac{1}{2}\left(  1-\nu\right)  +\frac{1}{2}\nu u\right)  }\,.
\end{align*}
The function $L(u)$ is that of (\ref{Lem}) and it is used to calculate the
2-particle form factor on the energy momentum (\ref{EMN}) and also to
calculate the 3-particle form factor of the field \ref{phi3'}. In particular
the 2-particle K-function of the scalar operator $\bar{\psi}\psi$ is up to a
constant equal to $K(u)$. With the normalization in (\ref{EM})%
\[
N_{2}^{\bar{\psi}\psi}=2m/\left(  \pi^{2}\nu^{2}c\Gamma^{2}\left(  \tfrac
{1}{2}\left(  1-\nu\right)  \right)  \right)
\]
we obtain (\ref{EMN})
\[
F_{\alpha_{1}\alpha_{2}}^{\bar{\psi}\psi}(\underline{\theta})=\mathbf{C}%
_{\alpha_{1}\alpha_{2}}\,\bar{v}(\theta_{1})u(\theta_{2})\,F_{0}(\theta_{12})
\]
which agrees with the result of \cite{KW}. The normalization is chosen such
that the form factor agrees for $\theta\rightarrow i\pi$ with the free field expression.

\section{Bound state form factors}

\label{siv}

We discuss the form factor equation (iv)%

\[
\operatorname*{Res}_{\theta_{12}=i\eta}F_{12\dots n}^{\mathcal{O}}(\theta
_{1},\theta_{2},\dots,\theta_{n})\,=F_{(12)\dots n}^{\mathcal{O}}%
(\theta_{(12)},\dots,\theta_{n})\,\sqrt{2}\Gamma_{12}^{(12)}%
\]
for the examples of section (\ref{s5}). Of course, one may easily calculate
the residues for two-particle form factors for the pseudo-potential
$J^{\alpha\beta}(x)$ (\ref{FJ}) and $\bar{\psi}\psi(x)$ (\ref{EMN}) directly,
however we will check here whether the general pinching procedure of appendix
\ref{sd} will give the same result. In addition we obtain the bound state form
factor of the three-particle form factor for the field.

\paragraph{Two-particle current form factor:}

By the form factor equation (iv) (\ref{1.16}) the two-particle bound state
form factor for the pseudo-potential $J^{\alpha\beta}(x)$ is%
\[
\operatorname*{Res}_{\theta_{12}=i\pi\nu}F_{12}^{J^{\alpha\beta}}(\theta
_{1},\theta_{2})\,=F_{(12)}^{J^{\alpha\beta}}(\theta_{(12)})\,\sqrt{2}%
\Gamma_{12}^{(12)}\,,~~\theta_{(12)}=\tfrac{1}{2}\left(  \theta_{1}+\theta
_{2}\right)
\]
where the bound state intertwiner $\Gamma_{12}^{(12)}$ is given by
(\ref{b1.40}) and (\ref{b1.43}).

In appendix \ref{sh} we calculated the two-particle form factor for the
pseudo-potential $J^{\alpha\beta}(x)$ in terms of the integral%
\begin{align*}
I(u_{12})  &  =\frac{1}{2\pi i}\int_{\mathcal{C}_{_{\underline{u}}}^{(o)}%
}dvI(\underline{u},\underline{v})\\
I(\underline{u},v)  &  =\tilde{h}(\underline{u},v)\tilde{\Psi}(\underline
{u},v),~\tilde{h}(\underline{u},v)=\tilde{\phi}(u_{1}-v)\tilde{\phi}%
(u_{2}-v),~\tilde{\Psi}(\underline{u},v)=\tilde{c}(u_{1}-v)\,
\end{align*}
with the result
\[
I(u)=\frac{2^{\nu}\pi\sqrt{\pi}\Gamma\left(  -\frac{1}{2}\nu\right)
\Gamma\left(  \frac{1}{2}\nu+\frac{1}{2}\right)  \cos\frac{1}{2}\pi\nu}%
{\sin\frac{1}{2}\pi\nu\left(  u+1\right)  \sin\frac{1}{2}\pi\nu\left(
u-1\right)  \Gamma\left(  \frac{1}{2}\nu+\frac{1}{2}\nu u\right)
\Gamma\left(  1+\frac{1}{2}\nu-\frac{1}{2}\nu u\right)  }%
\]
and the residue%
\[
\operatorname*{Res}_{u=1}I(u)=-\left(  \Gamma\left(  -\frac{1}{2}\nu\right)
\right)  ^{2}.
\]

In appendix \ref{sd} we remarked that the residue is obtained from pinching
at:\newline$z=\theta_{1}-i\pi\nu\approx\theta_{2}$ for $\mathcal{C}^{(o)}$
\begin{align*}
\operatorname*{Res}_{u_{12}=1}I(u_{12})  &  =\operatorname*{Res}_{u_{12}%
=1}\frac{1}{2\pi i}\oint\limits_{u_{1}-1}dv\,\tilde{\phi}(u_{1}-v)\tilde{\phi
}(u_{2}-v)\tilde{c}(u_{1}-v)\\
&  =\operatorname*{Res}_{u_{12}=1}\tilde{\phi}(1)\tilde{\phi}(-u_{12}%
+1)=-\left(  \Gamma\left(  -\frac{1}{2}\nu\right)  \right)  ^{2}%
\end{align*}
which means that the pinching procedure gives the same result as the direct calculation.

\paragraph{Two-particle form factor of $\bar{\psi}\psi$:}

By the form factor equation (iv) (\ref{1.16}) the two-particle bound state
form factor for $\bar{\psi}\psi$ is%
\[
\operatorname*{Res}_{\theta_{12}=i\pi\nu}F_{12}^{\bar{\psi}\psi}(\theta
_{1},\theta_{2})\,=F_{(12)}^{\bar{\psi}\psi}(\theta_{(12)})\,\sqrt{2}%
\Gamma_{12}^{(12)}\,.
\]
In appendix \ref{sf} we calculated the two-particle form factor for $\bar
{\psi}\psi(x)$ in terms of the integral%
\begin{align*}
J(u_{12})  &  =\frac{1}{(2\pi i)^{2}}\int_{\mathcal{C}_{\underline{u}}^{(o)}%
}dv_{1}\int_{\mathcal{C}_{\underline{u}}^{(e)}}dv_{2}\,J(\underline
{u},\underline{v})\\
J(\underline{u},\underline{v})  &  =\tilde{\phi}(u_{1}-v_{1})\tilde{c}%
(u_{1}-v_{1})\tilde{\phi}(u_{1}-v_{2})\tilde{c}(u_{1}-v_{2})\tilde{\phi}%
(u_{2}-v_{1})\tilde{\phi}(u_{2}-v_{2})\varphi(v_{12})\\
\varphi(v)  &  =\frac{\left(  1-v\right)  K(v,k+1)}{\tilde{\phi}(v)\tilde
{\phi}(-v)\left(  v+1/\nu-k-1\right)  }%
\end{align*}
with the result
\[
J=\frac{c_{2}}{\cos\pi\nu}K(u,0)=\frac{c_{2}}{\cos\pi\nu}\frac{\Gamma\left(
1-\frac{1}{2}\nu-\frac{1}{2}\nu u\right)  \Gamma\left(  -\frac{1}{2}\nu
+\frac{1}{2}\nu u\right)  }{\Gamma\left(  1+\frac{1}{2}-\frac{1}{2}\nu
u\right)  \Gamma\left(  \frac{1}{2}+\frac{1}{2}\nu u\right)  }%
\]
and the residue is%
\[
\operatorname*{Res}_{u=1}J(u)=\frac{c_{2}}{\cos\pi\nu}\operatorname*{Res}%
_{u=1}K(u,0)=\frac{4}{\left(  1-\nu\right)  \pi}\left(  \Gamma\left(
-\tfrac{1}{2}\nu\right)  \right)  ^{2}.
\]

In appendix \ref{sd} we remarked that the residue is obtained from pinching
at:\newline$z=\theta_{1}-i\pi\nu\approx\theta_{2}$ for $\mathcal{C}^{(o)}$ and
$z_{j}=\theta_{2}\approx\theta_{1}-i\pi\nu$ for $\mathcal{C}^{(e)}$, therefore
(see (\ref{Ju}))%
\begin{align*}
\operatorname*{Res}_{u_{12}=1}J  &  =\operatorname*{Res}_{u_{12}=1}\frac
{1}{(2\pi i)^{2}}\left(  ~\oint\limits_{u_{1}-1}dv_{1}\int_{\mathcal{C}%
_{\underline{u}}^{(e)}}dv_{2}-\int_{\mathcal{C}_{\underline{u}}^{(o)}}%
dv_{1}\oint\limits_{u_{2}}dv_{2}\right)  J(\underline{u},\underline{v}%
)=R_{1}+R_{2}\\
R_{1}  &  =-\operatorname*{Res}_{u_{12}=1}\sum\limits_{l_{2}=0}^{\infty
}\left(  s_{11}(u_{1},u_{2},0,l_{2})+s_{12}(u_{1},u_{2},0,l_{2})\right)  \,.
\end{align*}
with%
\begin{gather*}
s_{ij}(u_{1},u_{2},l_{1},l_{2})=\operatorname*{Res}_{v_{1}=v_{o}(u_{i},l_{1}%
)}\operatorname*{Res}_{v_{2}=v_{e(u_{j},l_{2})}}J(u_{1},u_{2};v_{1},v_{2})\\
v_{o}(u,l)=u-1+2l/\nu,~v_{e}(u,l)=u-2l/\nu
\end{gather*}
It turns out that $s_{12}$ gives no contribution and
\[
R_{1}=-\operatorname*{Res}_{u_{12}=1}\sum\limits_{l_{2}=0}^{\infty}%
s_{11}(u_{1},u_{2},0,l_{2})=2\frac{\left(  \Gamma\left(  -\frac{1}{2}%
\nu\right)  \right)  ^{2}}{\pi\left(  1-\nu\right)  }%
\]
such that again%
\[
\operatorname*{Res}_{u_{12}=1}J=4\frac{\left(  \Gamma\left(  -\frac{1}{2}%
\nu\right)  \right)  ^{2}}{\pi\left(  1-\nu\right)  }%
\]
which means that the pinching procedure gives the same result as the direct calculation.

\paragraph{3-particle form factor of $\psi$:}

We discuss the bound state fusion of 2 fundamental fermions $f+f\rightarrow
b_{2}$. We write (\ref{chi}) as
\[
\psi(x)=(i\gamma\partial+m)\tilde{\chi}{(}x{),~~}\tilde{\chi}{(}x{)=-i}\left(
\square+m^{2}\right)  ^{-1}\chi{(}x{)}%
\]
and apply the form factor equation (\ref{1.16}) to\footnote{Strictly
 speaking $F_{1\bar{1}1}^{\tilde{\chi}}\pm F_{\bar{1}%
11}^{\tilde{\chi}}$ give $F_{b_{2}^{(0,2)}1}^{\tilde{\chi}}$.} $\tilde{\chi}$
\[
\operatorname*{Res}\limits_{\theta_{12}=i\pi\nu}F_{1\bar{1}1}^{\tilde{\chi}%
}(\underline{\theta})=F_{b_{2}1}^{\tilde{\chi}^{(\pm)}}(\theta_{0},\theta
_{3})\sqrt{2}\Gamma_{1\bar{1}}^{b_{2}}\,.
\]
The component $K_{1\bar{1}1}^{\tilde{\chi}}$ of the K-function (similar as for
$\bar{\psi}\psi$ in appendix \ref{sf}) can be written in terms of
\begin{align*}
J^{\chi}(\underline{u})  &  =\frac{1}{(2\pi i)^{2}}\int_{\mathcal{C}%
_{\underline{u}}^{(o)}}dv_{1}\int_{\mathcal{C}_{\underline{u}}^{(e)}}%
dv_{2}\,J^{\chi}(\underline{u},\underline{v})p^{\chi}(\underline{u}%
,\underline{v})\\
J^{\chi}(\underline{u},\underline{v})  &  =\left(  \prod_{i=1}^{3}\prod
_{j=1}^{2}\tilde{\phi}(u_{i}-v_{j})\right)  \tilde{b}(u_{1}-v_{1})\tilde
{b}(u_{1}-v_{2})\tilde{c}(u_{2}-v_{1})\tilde{c}(u_{2}-v_{2})\varphi(v_{12})\\
\varphi(v)  &  =\frac{\left(  1-v\right)  K(v,1)}{\tilde{\phi}(v)\tilde{\phi
}(-v)\left(  v+1/\nu-1\right)  }\,.
\end{align*}
In appendix \ref{sd} we remarked that the residue is obtained from pinching
at:\newline$z_{1}=\theta_{1}-i\pi\nu\approx\theta_{2}~(v_{1}=u_{1}-1\approx
u_{2})$ for $\mathcal{C}^{(o)}$ and $z_{2}=\theta_{2}\approx\theta_{1}-i\pi
\nu~(v_{2}=u_{2}\approx u_{1}-1)$ for $\mathcal{C}^{(e)}$. Therefore the bound
state form factor is obtained from%
\[
\operatorname*{Res}_{u_{12}=1}J^{\chi}(\underline{u})=\operatorname*{Res}%
_{u_{12}=1}\frac{1}{(2\pi i)^{2}}\left(  \oint_{u_{1}-1}\int_{\mathcal{C}%
_{\underline{u}}^{(e)}}-\int_{\mathcal{C}_{\underline{u}}^{(o)}}\oint_{u_{2}%
}\right)  dv_{1}\,dv_{2}J^{\chi}(\underline{u},\underline{v})p^{\chi
}(\underline{u},\underline{v})\,.
\]
The integrals may be calculated in terms of hypergeometric functions
$_{3}F_{2}$. We obtain%
\begin{multline*}
F_{b_{2}1}^{\tilde{\chi}^{(\pm)}}(\theta_{0},\theta_{3})=e^{\mp\frac{1}%
{2}\theta_{0}}\left(  e^{\pm\frac{1}{4}i\pi\nu}f_{13}(\theta_{03})+e^{\mp
\frac{1}{4}i\pi\nu}f_{32}(\theta_{03})\right) \\
+e^{\mp\frac{1}{2}\theta_{3}}\left(  e^{\pm\frac{1}{2}i\pi\nu}f_{11}%
(\theta_{03})+e^{\mp\frac{1}{2}i\pi\nu}f_{22}(\theta_{03})\right)
\end{multline*}
where $f_{1i}(\theta_{03})$ and $f_{i2}(\theta_{03})$ are the results from the
integrations%
\[
\oint_{u_{1}-1}dv_{1}\int_{\mathcal{C}_{u_{i}}}\,dv_{2}..\dots~\text{and
}\oint_{\mathcal{C}_{u_{i}}}dv_{1}\int_{u_{2}}\,dv_{2}..\dots
~\text{respectively.}%
\]
For example up to a constant (see Fig. \ref{f13a})$\,$%
\begin{align*}
f_{13}(u)  &  =\frac{\left(  \Gamma\left(  1-\frac{1}{4}\nu-\frac{1}{2}\nu
u\right)  \right)  ^{2}\Gamma\left(  -\frac{3}{4}\nu+\frac{1}{2}\nu u\right)
\Gamma\left(  -\frac{1}{4}\nu+\frac{1}{2}\nu u\right)  }{\Gamma\left(
\frac{3}{2}-\frac{3}{4}\nu+\frac{1}{2}\nu u\right)  \Gamma\left(  \frac{3}%
{2}-\frac{1}{4}\nu-\frac{1}{2}\nu u\right)  \cot\frac{1}{2}\pi\nu\left(
u-\frac{1}{2}\right)  \cot\frac{1}{2}\pi\nu\left(  u+\frac{1}{2}\right)  }\\
&  \times\,_{3}F_{2}\left(  -\tfrac{1}{2}\nu+1,-\tfrac{3}{4}\nu+\tfrac{1}%
{2}\nu u,-\tfrac{1}{2}+\tfrac{1}{4}\nu+\tfrac{1}{2}\nu u;\tfrac{1}{4}%
\nu+\tfrac{1}{2}\nu u,\tfrac{3}{2}-\tfrac{3}{4}\nu+\tfrac{1}{2}\nu u;1\right)
F_{b}(u)
\end{align*}
where $F_{b}(u)$ is the minimal highest weight form factor function in the
$b_{2}^{(r)}+f$ sector
\[
F_{b}(\theta)=const.\left(  \sinh\tfrac{1}{2}\theta\right)  \frac{F_{+}^{\min
}(\theta+\tfrac{1}{2}i\pi\nu)F_{+}^{\min}(\theta-\tfrac{1}{2}i\pi\nu)}%
{\Gamma(1+\tfrac{1}{4}\nu-\frac{\theta}{2i\pi})\Gamma(\tfrac{1}{4}\nu
+\frac{\theta}{2i\pi})}%
\]
or explicitly in terms of $G\left(  z\right)  $ Barnes G-function%
\[
F_{b}(u)=\frac{\left(  \sinh\frac{1}{2}\theta\right)  G\left(  \frac{1}{4}%
\nu+\frac{1}{2}\nu u\right)  G\left(  \frac{3}{2}-\frac{3}{4}\nu-\frac{1}%
{2}\nu u\right)  G\left(  1+\frac{1}{4}\nu-\frac{1}{2}\nu u\right)  G\left(
\frac{1}{2}-\frac{3}{4}\nu+\frac{1}{2}\nu u\right)  }{G\left(  \frac{1}%
{2}+\frac{1}{4}\nu+\frac{1}{2}\nu u\right)  G\left(  2-\frac{3}{4}\nu-\frac
{1}{2}\nu u\right)  G\left(  \frac{3}{2}+\frac{1}{4}\nu-\frac{1}{2}\nu
u\right)  G\left(  1-\frac{3}{4}\nu+\frac{1}{2}\nu u\right)  }%
\]
with $u=\theta/(i\pi\nu)$. It satisfies Watson's equation%
\[
\frac{F_{b}(\theta)}{F_{b}(-\theta)}=a(\theta+\tfrac{1}{2}i\pi\nu
)a(\theta-\tfrac{1}{2}i\pi\nu)\frac{\theta+\tfrac{1}{2}i\pi\nu}{\theta
-\tfrac{1}{2}i\pi\nu}=a_{b}(\theta)
\]
where $a_{b}(\theta)$ is the highest weight scattering amplitude in the
$b_{2}^{(r)}+f$ sector.

\section{1/N expansion}

\subsection{1/N expansion of the exact 3-particle field form factor}

\label{sg}

For $\chi^{\delta}{(}x{)}=i(-i\gamma\partial+m)\psi^{\delta}(x)$ we derive for
the highest weight component $\chi{(}x{)=}\chi^{1}{(}x{)}$%
\begin{equation}
F_{\bar{1}11}^{\chi}(\underline{\theta})=\frac{8\pi m}{N}\,\left(  \frac
{\cosh\frac{1}{2}\theta_{12}}{\theta_{12}-i\pi}\,u(\theta_{3})-\frac
{\cosh\frac{1}{2}\theta_{13}}{\theta_{13}-i\pi}\,u(\theta_{2})\right)
+O(N^{-2}) \label{F3}%
\end{equation}
which is equivalent to (\ref{F3g}).

\begin{proof}
The p-function of $\chi{(}x{)}$ for three particles and $\nu=0$ is%
\[
p^{\chi^{(\pm)}}=\exp\left(  \mp\tfrac{1}{2}\left(  \theta_{1}+\theta
_{2}+\theta_{3}-z_{1}-z_{2}\right)  \right)  .
\]
We have to consider (up to const.)
\[
K_{\bar{1}11}^{\chi^{(+)}}(\underline{\theta})=\int_{\mathcal{C}%
_{\underline{\theta}}}dz_{1}\int_{\mathcal{C}_{\underline{\theta}}}%
dz_{2}\,\prod_{i=1}^{3}\left(  \tilde{\phi}(\theta_{i}-z_{1})\tilde{\phi
}(\theta_{i}-z_{2})\right)  \frac{1}{\tilde{\phi}(z_{12})\tilde{\phi}%
(-z_{12})}p^{\chi^{(+)}}(\underline{z})\,\tilde{\Psi}_{\bar{1}11}%
(\underline{\theta},\underline{z}).
\]
This formula is similar as (\ref{rec}) for $k=0$ (which correspond to the
operator $\bar{\psi}\psi$), only we have here to add the factor $\left(
\tilde{\phi}(\theta_{3}-z_{1})\tilde{\phi}(\theta_{3}-z_{2})p^{\chi^{(\pm)}%
}(\underline{\theta},\underline{z})\right)  $. Therefore we get using
(\ref{rec1}) for small $\nu$ (up to constants)%
\begin{align*}
K_{\bar{1}11}^{\chi^{(\pm)}}(\underline{\theta})  &  =K(\theta_{12}%
,0)\frac{\exp\left(  \mp\tfrac{1}{2}\theta_{3}\right)  }{\sinh\frac{1}%
{2}\theta_{13}\sinh\frac{1}{2}\theta_{23}}+(2\leftrightarrow3)\\
&  =\frac{\cosh\frac{1}{2}\theta_{12}}{\left(  \theta_{12}-i\pi\right)
\sinh\frac{1}{2}\theta_{12}}\frac{\exp\left(  \mp\tfrac{1}{2}\theta
_{3}\right)  }{\sinh\frac{1}{2}\theta_{13}\sinh\frac{1}{2}\theta_{23}%
}+(2\leftrightarrow3)+O(\nu)\\
&  =\frac{1}{\theta_{12}-i\pi}\coth\tfrac{1}{2}\theta_{12}\frac{\exp\left(
\mp\tfrac{1}{2}\theta_{3}\right)  }{\sinh\frac{1}{2}\theta_{13}\sinh\frac
{1}{2}\theta_{23}}+(2\leftrightarrow3)+O(\nu)\\
F_{\bar{1}11}^{\chi}(\underline{\theta})  &  =\frac{\cosh\frac{1}{2}%
\theta_{12}}{\theta_{12}-i\pi}\,u(\theta_{3})-\frac{\cosh\frac{1}{2}%
\theta_{13}}{\theta_{13}-i\pi}\,u(\theta_{2})+O(\nu)\end{align*}
which is (\ref{F3}) up to a constant. The normalization is obtained by the
form factor equation (iii)%
\begin{align*}
\operatorname*{Res}_{\theta_{12}=i\pi}F_{\bar{1}11}^{\psi}(\underline{\theta
})  &  =2i\left(  1-a(\theta_{23})\right)  F_{1}^{\psi}(\theta_{3})\\
&  =\frac{4\pi}{N}\left(  \frac{1}{\sinh\theta_{23}}-\frac{1}{\theta_{23}%
}\right)  u(\theta_{3})+O(N^{-2})
\end{align*}
where
\[
F_{\alpha\beta\gamma}^{\psi}(\underline{\theta})=\frac{i(\gamma\left(
p_{1}+p_{2}+p_{3}\right)  +m)}{8m^{2}\cosh\frac{1}{2}\theta_{12}\cosh\frac
{1}{2}\theta_{13}\cosh\frac{1}{2}\theta_{23}}F_{\alpha\beta\gamma}^{\psi\chi
}(\underline{\theta})\,.
\]
It has been used that
\begin{align*}
K(\theta,0)  &  =-2i\pi\frac{\cosh\frac{1}{2}z}{\left(  z-i\pi\right)
\sinh\frac{1}{2}z}+O(\nu)\\
\tilde{\phi}(\theta)  &  =\frac{-i\pi}{\sinh\frac{1}{2}\theta}+O(\nu)\\
F(\theta)  &  =-i\sinh\tfrac{1}{2}\theta+O\left(  \nu\right) \\
a(\theta)  &  =1+\nu i\pi\left(  \frac{1}{\sinh\theta}-\frac{1}{\theta
}\right)  \,+O\left(  \nu^{2}\right)  .
\end{align*}

\end{proof}

\subsection{1/N perturbation theory}

\label{s1overN}

Introducing the auxiliary field $\sigma(x)$ the Lagrangian (\ref{LGN}) may be
written as%
\[
\mathcal{L}^{GN}=\bar{\psi}(i\gamma\partial-\sigma)\psi-\frac{1}{2g^{2}}%
\sigma^{2}%
\]
and the Green's functions in $1/N$ expansions are obtained from the expansion
of%
\begin{align*}
Z(\xi,\bar{\xi})  &  =\int d\sigma\,\exp\left(  i\mathcal{A}_{eff}%
(\sigma)-\bar{\xi}S\xi\right) \\
\mathcal{A}_{eff}(\sigma)  &  =-i\tfrac{1}{2}N\,\operatorname*{tr}%
\,\ln(i\gamma\partial-\sigma)-\int d^{2}x\,\frac{1}{2g^{2}}\sigma^{2}%
\end{align*}
with the $\sigma$ propagator \cite{ZZ4,KW}%
\[
\tilde{\Delta}_{\sigma}(k)=\left(  \tfrac{1}{2}N\int\frac{d^{2}p}{(2\pi)^{2}%
}\operatorname*{tr}\left(  \frac{1}{\gamma p-m}\left(  \frac{1}{\gamma
(p+k)-m}-\frac{1}{m}\right)  \right)  \,\right)  ^{-1}=-\frac{4\pi i}%
{N}\,\frac{\tanh\frac{1}{2}\phi}{\phi}%
\]
where $k^{2}=-4m^{2}\sinh^{2}\frac{1}{2}\phi$. This propagator together with
the simple vertex of Fig.~\ref{fa1a} \begin{figure}[tbh]%
\[%
\begin{array}
[c]{c}%
\unitlength4mm\begin{picture}(3,4) \put(0,0){\line(0,1){2}} \put(0,0){\vector(0,1){1}} \put(0,2){\line(0,1){2}} \put(0,2){\vector(0,1){1}} \put(0,2){\dashbox{.2}(2,0){}} \put(.5,1){$k$} \put(.5,2.3){$\leftarrow$} \put(2.2,1.8){$\sigma$} \end{picture}
\end{array}
=-i
\]
\caption{The elementary vertex for the $O(N)$ Gross-Neveu model. With respect
to isospin the vertex is proportional to the unit matrix\textit{. }}%
\label{fa1a}%
\end{figure}
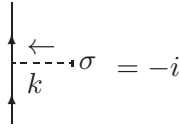yield the Feynman rules which allow to calculate general vertex
functions in the $1/N$-expansion. For example the four point vertex function
is
\begin{equation}
\tilde{\Gamma}^{(4)}{}_{AB\alpha\beta}^{DC\delta\gamma}(-p_{3},-p_{4}%
,p_{1},p_{2})=\delta_{\alpha}^{\delta}\delta_{\beta}^{\gamma}\,G_{AB}%
^{DC}(p_{2}-p_{3})-\delta_{\alpha}^{\gamma}\delta_{\beta}^{\delta}%
\,G_{AB}^{CD}(p_{3}-p_{1}) \label{a.25a}%
\end{equation}
where $A,B,C,D$ are spinor indices, $\alpha\beta\gamma\delta$ are isospin
indices and $G$ is given by the Feynman graph of Fig.~\ref{fa2}.
\begin{figure}[tbh]%
\[
G_{AB}^{DC}(k)=
\begin{array}
[c]{c}%
\unitlength5mm
\begin{picture}(8,6) \put(1,1){\line(1,2){1}} \put(1,1){\vector(1,2){.5}} \put(2,3){\line(-1,2){1}} \put(2,3){\vector(-1,2){.5}} \put(7,1){\line(-1,2){1}} \put(7,1){\vector(-1,2){.5}} \put(6,3){\line(1,2){1}} \put(6,3){\vector(1,2){.5}} \put(2,3){\dashbox{.2}(4,0){}} \put(3.7,2){$k$} \put(3.5,3.3){$\leftarrow$} \put(.5,.1){$A$} \put(7,.1){$B$} \put(7,5.3){$C$} \put(.5,5.3){$D$} \put(1.8,1.5){$p_1$} \put(5.6,1.5){$p_2$} \put(5.6,4.5){$p_3$} \put(1.8,4.5){$p_4$} \end{picture}
\end{array}
\]
\caption{The four point vertex }%
\label{fa2}%
\end{figure}
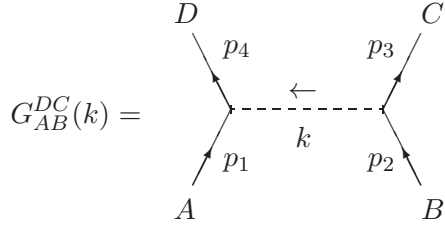Taking into account the contributions from the propagator we
obtain
\begin{equation}
G(k)=-1\otimes1\,\tilde{\Delta}_{\sigma}(k)=\frac{4\pi i}{N}1\otimes
1\,\frac{\tanh\frac{1}{2}\phi}{\phi}\,. \label{a.30a}%
\end{equation}
where the tensor product structure of the spinor matrices is obvious from
Fig.~\ref{fa2}.

\paragraph{3-particle form factor of the fundamental fermi field:}

We now calculate the three particle form factor of the fundamental fermi field
in $1/N$-expansion in lowest nontrivial order. For convenience we multiply the
field with the Dirac operator%
\[
\chi{^{\delta D}(}x{)}=i(-i\gamma\partial+m)_{D^{\prime}}^{D}\psi^{\delta
D^{\prime}}(x)
\]
and define
\[
_{out}^{~~\gamma}\langle\,p_{3}\,|\,\chi{^{\delta D}}(0)\,|\,\theta_{1}%
,\theta_{2}\,\rangle_{\alpha\beta}^{in}={F^{\eta{^{\delta D}}}\,}_{\alpha
\beta}^{\gamma}(\theta_{3};\theta_{1},\theta_{2})\,.
\]
By means of LSZ-techniques one can express the connected part in terms of the
4-point vertex function (\ref{a.25a}) in lowest order given by the Feynman
graphs of Fig.~\ref{fa3a}%
\begin{equation}
{F_{{conn}}^{\chi{^{\delta D}}}}_{\alpha\beta}^{\gamma}(\theta_{3};\theta
_{1},\theta_{2})=\bar{u}_{C}(p_{3})\left\{  \delta_{\alpha\delta}\delta
_{\beta\gamma}\,G_{AB}^{DC}(p_{2}-p_{3})-\delta_{\alpha\gamma}\delta
_{\beta\delta}\,G_{AB}^{CD}(p_{3}-p_{1})\right\}  u^{A}(p_{1})u^{B}%
(p_{2})\label{a.25b}%
\end{equation}
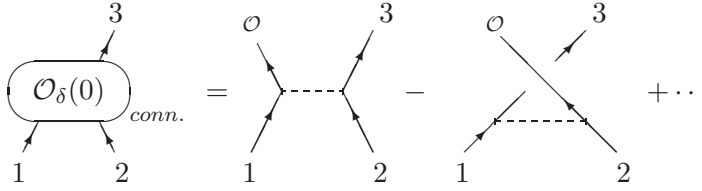
\begin{figure}[tbh]%
\[%
\begin{array}
[c]{c}%
\unitlength4mm
\begin{picture}(25,5) \put(2,3){\oval(4,2)} \put(2,3){\makebox(0,0){${\cal O}_\delta(0)$}} \put(.5,1){\line(1,2){.5}} \put(.5,1){\vector(1,2){.3}} \put(3.5,1){\line(-1,2){.5}} \put(3.5,1){\vector(-1,2){.3}} \put(3,4){\line(1,2){.5}} \put(3,4){\vector(1,2){.3}} \put(4,2){$\scriptstyle conn.$} \put(.1,0){$1$} \put(3.5,0){$2$} \put(3.3,5.3){$3$} \put(6.5,2.7){$=$} \put(8,1){\line(1,2){1}} \put(8,1){\vector(1,2){.7}} \put(9,3){\line(-1,2){.8}} \put(9,3){\vector(-1,2){.5}} \put(7.7,5){$\scriptstyle{\cal O}$} \put(9,3){\dashbox{.2}(2,0){}} \put(12,1){\line(-1,2){1}} \put(12,1){\vector(-1,2){.7}} \put(11,3){\line(1,2){1}} \put(11,3){\vector(1,2){.7}} \put(7.7,0){$1$} \put(12,0){$2$} \put(12.2,5.3){$3$} \put(13,2.7){$-$} \put(15,1){\line(1,1){2}} \put(15,1){\vector(1,1){.7}} \put(18,4){\line(1,1){1}} \put(18,4){\vector(1,1){.6}} \put(15.7,5){$\scriptstyle{\cal O}$} \put(16,2){\dashbox{.2}(3,0){}} \put(20,1){\line(-1,1){3.8}} \put(20,1){\vector(-1,1){1.7}} \put(14.7,0){$1$} \put(20,0){$2$} \put(19.2,5.3){$3$} \put(21,2.7){$+\cdots$} \end{picture}
\end{array}
\]
\caption{The connected part of the three particle form factor of the
fundamental fermi field in $1/N$-expansion. }%
\label{fa3a}%
\end{figure}where $G$ is given by Fig.~\ref{fa2} and eq.~(\ref{a.30a}) and the
spinors by $u_{\pm}(p)=\sqrt{m}e^{\mp\theta/2}$. It turns out that for
$p_{1},~p_{2}$ and $p_{3}$ on-shell several terms vanish or cancel and we
obtain up to order $1/N$ using $\bar{u}(\theta_{1})u(\theta_{2})=2m\cosh
\tfrac{1}{2}\theta_{12}$%
\begin{equation}
{F_{{conn}.}^{\chi{^{\delta D}}}}_{\alpha\beta}^{\gamma}=\frac{i\pi}%
{N}\,8m\,\bigg\{\delta_{\alpha}^{\delta}\delta_{\beta}^{\gamma}\frac
{\sinh\frac{1}{2}\theta_{23}}{\theta_{23}}\,u^{D}(p_{1})-\delta_{\alpha
}^{\gamma}\delta_{\beta}^{\delta}\frac{\sinh\frac{1}{2}\theta_{13}}%
{\theta_{13}}\,u^{D}(p_{2})\bigg\}\,.\label{a.40a}%
\end{equation}
By crossing ($\theta_{3}\rightarrow\theta_{3}+i\pi$) this gives ${F_{\alpha
\beta\bar{\gamma}}^{\chi^{\delta}}}$ and agrees with the $1/N$ expansion of
the exact form factor (\ref{F3g}).

\providecommand{\href}[2]{#2}\begingroup\raggedright\endgroup

\end{document}